\documentclass[12pt,a4paper,final]{iopart}

\usepackage{iopams}  
\usepackage{graphicx}
\usepackage[breaklinks=true,colorlinks=true,linkcolor=blue,urlcolor=blue,citecolor=blue]{hyperref}

\usepackage{subfigure}
\usepackage{amstext,amsthm,bm}
\newtheorem{thm}{Theorem}[section]
\newtheorem{defn}[thm]{Definition}
\newtheorem{coro}[thm]{Corollary}
\newtheorem{prop}[thm]{Proposition}

\newcommand{\trn}[1]{\text{Tr}_{#1}}
 
\newcommand{\ket}[1]{| #1 \rangle}
\newcommand{\bra}[1]{\langle #1 |}

\newcommand{\inprod}[2]{\bra{#1}#2\rangle}
\newcommand{\ketbra}[1]{\ket{#1}\bra{#1}} 
\newcommand{\sqket}[1]{|\!#1 \rangle}

\newcommand{\exval}[2]{\langle #1 \rangle_{#2}}

\newcommand{\identity}{\mathbb{I}}
\newcommand{\hilbert}{\mathcal{H}}
\newcommand{\map}{\mathcal{M}}

\usepackage{tikz}

\begin{document}


\title[]{On state vs. channel quantum extension problems: 
exact results for {\em U}$\,\otimes${\em U}$\,\otimes${\em U} symmetry}

\author{Peter D. Johnson}
\address{Department of Physics and Astronomy, Dartmouth 
College, 6127 Wilder Laboratory, Hanover, NH 03755, USA}
\ead{Peter.D.Johnson.GR@Dartmouth.edu}

\author{Lorenza Viola}
\address{Department of Physics and Astronomy, Dartmouth 
College, 6127 Wilder Laboratory, Hanover, NH 03755, USA}
\ead{Lorenza.Viola@Dartmouth.edu}

\begin{abstract}
We develop a framework which unifies seemingly different extension (or ``joinability'') problems for 
bipartite quantum states and channels. This includes well known extension problems such as optimal quantum cloning 
and quantum marginal problems as special instances. Central to our generalization is a variant of the Choi-Jamiolkowski isomorphism between bipartite states and dynamical maps which we term the ``homocorrelation map": while the former emphasizes the preservation of the positivity constraint, the latter is designed to preserve {\em statistical correlations}, allowing direct contact with entanglement. 
In particular, we define and  analyze state-joining, channel-joining, and local-positive joining problems in three-party settings exhibiting collective $U\otimes U\otimes U$ symmetry, obtaining exact analytical characterizations in low dimension. Suggestively, we find that bipartite 
quantum states are limited in the degree to which their measurement  outcomes may agree, while 
quantum channels are limited in the degree to which their measurement outcomes may disagree. Loosely speaking, quantum mechanics enforces an upper bound on the extent of {\em positive} correlation across a bipartite system at a given time, as well as on the extent of {\em negative} correlation between the state of a same system across two 
instants of time.  We argue that these general statistical bounds inform the quantum joinability limitations, and show 
that they are in fact sufficient for the three-party $U\otimes U\otimes U$-invariant setting.

\end{abstract}

\vspace{2pc}

\section{Introduction}
\label{sec:intro}

It has long been appreciated that many of the intuitive features of classical probability theory do not translate to quantum theory. For instance, every classical probability distribution has a unique decomposition into extremal distributions, whereas a general density operator does not admit a unique decomposition in terms of extremal operators (pure states). Entanglement is responsible for another distinctive trait of quantum theory:  as vividly expressed by Schr\"{o}dinger back in 1935 \cite{Schroedinger1935}, ``the best possible knowledge of a total system does not necessarily include total knowledge of all its parts,'' in striking contrast to the classical case. Certain features of classical probability theory do, nonetheless, carry over to the quantum domain. While it is natural to view these distinguishing features as a consequence of quantum theory being a non-commutative generalization of classical probability theory in an appropriate sense, thoroughly understanding {\em how} and {\em the extent to 
which} 
the purely quantum features of the theory arise from its mathematical structure remains a longstanding central question across quantum foundations, mathematical physics, and quantum information processing (QIP), see e.g. Refs. 
\cite{vonNeumann1955,Accardi1990,Leifer2006,Barnum2012}. 

In this paper, we investigate a QIP-motivated setting which allows us to directly compare and contrast features of
quantum theory with classical probability theory, namely, the relationship between the parts (subsystems) of a composite quantum system and the system as a whole.
Specifically, building on our earlier work \cite{Johnson2013}, we develop and investigate a general framework for what we refer to as \emph{quantum joinability}, which addresses the compatibility of different statistical correlations among quantum measurements on different systems. Arguably, 
the most familiar case of joinability is provided by the ``quantum marginal'' (aka ``local consistency'') problem \cite{Klyachko2006,Liu2006}. In this case, we ask whether there exists a joint quantum state compatible with a given set of reduced states on (typically non-disjoint) groupings of subsystems. The quintessential example of a failure of joinability is the fact that two pairs of two-level systems (qubits), say, Alice-Bob ($A$-$B$) and Alice-Charlie ($A$-$C$), cannot simultaneously be described by the singlet state, $\ket{\psi^-}=\sqrt{1/2}(\sqket{\uparrow \downarrow} - \sqket{\downarrow \uparrow})$. A seminal exploration of this observation was carried out by Coffman, Kundu, and Wootters \cite{Wootters2000} and later dubbed the ``monogamy of entanglement'' \cite{Terhal2004}. In classical probability theory, a necessary and sufficient condition for marginal probability distributions on $A$-$B$ and $A$-$C$ to admit a joint probability distribution (or ``extension") on $A$-$B$-$C$ is that the marginals 
over $A$ 
be equal \cite{Klyachko2006,Fritz2013}. The analogous compatibility condition remains necessary in quantum theory, but, as demonstrated by the above example, is clearly no longer sufficient. The identification of necessary and sufficient conditions in general settings with overlapping marginals remains an actively investigated open problem as yet \cite{Johnson2013,Lieb2013,Chen2013}.

Physically, standard {\em state-joinability problems} as formulated above for density operators, may be regarded as characterizing the compatibility of statistical correlations of two (or more) different subsystems at a {\em given} time. However, correlations between the same system {\em before and after} the action of a quantum channel -- a completely positive trace-preserving (CPTP) dynamical map -- may also be considered, for example, in order to characterize the ``location'' of quantum information that one subsystem may carry about another \cite{Griffiths2005} and/or the causal structure of the events on which probabilities are defined \cite{Leifer2006}. With this in mind, one may formulate an analog quantum marginal problem for quantum channels (see also Ref. \cite{Chen2012}). For example, given two quantum channels $\map_{AB}:\mathcal{B}(\hilbert_A)\rightarrow\mathcal{B}(\hilbert_B)$ and $\map_{AC}:\mathcal{B}(\hilbert_A)\rightarrow\mathcal{B}(\hilbert_C)$ (where we notate the space of bounded linear 
maps on a Hilbert space $\hilbert$ with
 $\mathcal{B}(\hilbert)$), one may ask whether there exists a quantum channel $\map_{
ABC}:\mathcal{B}(\hilbert_A)\rightarrow\mathcal{B}(\hilbert_B\otimes \hilbert_C)$, whose reduced channels are $\map_{AB}$ and $\map_{AC}$, respectively.

A motivation for considering such \emph{channel-joinability problems} is that questions regarding the optimality of paradigmatic QIP tasks such as quantum cloning \cite{Cerf2000,Iblisdir2005} or broadcasting \cite{Ghiu2003} may be naturally recast as such. A fundamental tool here is the {\em Choi-Jamiolkowski isomorphism} \cite{Jamiolkowski1972,Choi1975}, which may been used to translate optimal cloning problems into quantum marginal problems \cite{Ramanathan2009,Horodecki2012}, and vice-versa \cite{Johnson2013}. Both monogamy of entanglement and the no-cloning theorem \cite{Wootters1982} have significant implications for the behavior of quantum systems: the former effectively constrains the kinematics of a multipartite quantum system, while the latter constrains the dynamics of a quantum system (composite or not). As both of these fundamental concepts are closely related to respective quantum joinability problems, we are prompted to explore in more depth their similarities and differences. Identifying a 
general joinability framework, able to encompass {\em all} such quantum marginal problems, is one of our main aims here.

%
The content is organized as follows. In Section \ref{sec:unifyingframework}, we introduce and motivate the use of what we term the \emph{homocorrelation map} as our main tool for representing quantum channels as bipartite operators. 
We formally define a notion of quantum joinability that incorporates all joinability problems of interest, and 
discuss ways in which different joinability problems may be (homomorphically) mapped into one another.
In Section \ref{sec:tripartite}, we obtain a complete analytical characterization of some archetypal examples of low-dimensional quantum joinability problems. Namely, we address three-party joinability of quantum states, quantum channels, and block-positive (or ``local-positive'') operators, in the case that the relevant operators are invariant under the group of collective unitary transformations, that is, under the action of arbitrary transformations of the form $U \otimes U\otimes U$. These examples allow us to distinguish the joinability limitations stemming from classical probability theory from those due to quantum theory and, furthermore, to contrast the joinability properties of quantum channels vs. states. In Section \ref{sec:agreement}, we investigate a possible source for the stricter joinability bounds in quantum theory, as compared to classical probability theory. We introduce the notion of \emph{degree of agreement} (disagreement), that is, the probability 
that a random local collective 
measurement yields same (different) outcomes, as given by an appropriate two-value POVM. 
We find that {\em quantum theory places different bounds on the degree of agreement arising from quantum states than it does on that of quantum channels}: while quantum states are limited in their degree of agreement, quantum channels are limited in their degree of disagreement. The differences in these bounds point to a crucial distinction between quantum channels and states. At least in the examples of Section \ref{sec:tripartite} and a few others, these limitations 
suffice in fact to determine the bounds of joinability {\em exactly}. Possible implications of such bounds with regards to joinability properties of general quantum states and channels are also discussed, and final remarks conclude in Section \ref{sec:conclusion}.

\section{General quantum joinability framework}
\label{sec:unifyingframework}

We begin by reviewing the standard state-joinability (quantum marginal) problem, framing it in a language suitable for generalization. Given a composite Hilbert space $\hilbert^{(N)}=\bigotimes_{i=1}^N\hilbert_i$, a \emph{joinability scenario} is defined by a list of partial traces $\{\trn{\ell_k}\}$, with each $\ell_k\subseteq [1,\ldots,N]$, along with a set of allowed ``joining operators,'' $W$, which in this case is the set of positive trace-one operators acting on $\hilbert^{(N)}$; accordingly, we may associate a joinability scenario with a 2-tuple $(W,\{\trn{\ell_k}\})$. For a given joinability scenario, the images of 
$W$ under the $\trn{\ell_k}$ define 
a set of reduced states $\{R_k\}=\{\trn{k}(W)\}$. For any list of states 
$\{\rho_k\}\in\{R_k\}$, the following definition then applies: 
\begin{defn}
{\bf [State-Joinability]} Given a joinability scenario described by the pair 
$(W \equiv \{w|\,\,w\geq 0\},\{\trn{\ell_k}\})$, the reduced states 
$\{\rho_k\}\in\{R_k\}$ are \emph{joinable} if there exists a 
joining state $w\in W$ such that $\trn{\ell_k}(w) =\rho_k$ for all $k$.
\end{defn} 


The first step toward achieving the intended generalization of the above definition to quantum channels is to represent the latter as bipartite operators. In the following subsection, we establish a tool to achieve this and highlight its broader utility.

\subsection{Homocorrelation map and positive cones}
\label{subsec:homocorr}

%
One way to identify channels with bipartite operators is by use of the Choi-Jamiolkowski (CJ) isomorphism  \cite{Choi1975,Jiang2013}. 
This isomorphism, denoted $\mathcal{J}$, identifies each map $\map\in\mathcal{L}(\hilbert_A,\hilbert_B)$ with the state resulting from its (the map's) action on one member of a Bell state:
\begin{equation}
\mathcal{J}(\map)\equiv [\mathcal{I}_A\otimes\map] (\ketbra{\Phi^+} )=
 \frac{1}{d_A}\sum_{ij} \ket{i}\bra{j} \otimes {\mathcal M}( \ket{i}\bra{j}),
 \label{choi}
\end{equation}
where $\mathcal{I}_A$ is the identity map on $\mathcal{B}(\hilbert_A)$ and $d_A=\text{dim}(\hilbert_A)$. We note that $d_{A}\ketbra{\Phi^+}=V^{T_A}$, where $V$ is the swap operator on $\hilbert_A \otimes \hilbert_A$ and $T_A$ denotes partial transposition on subsystem $A$. 
The transformation is an isomorphism in that it preserves the positivity of the objects it maps to and from;
namely, quantum channels (CPTP maps) are mapped to quantum states (positive trace-one operators). Consequently, the CJ isomorphism is a useful diagnostic tool for determining whether or not a map is CP.
It does depend on a choice of local basis (to define $\ket{\Phi^+}$ and $T_A$).
For the isomorphism to hold, the reference state ($\ketbra{\Phi^+}$ above) must be maximally entangled; and, for $d>2$, any such state reflects a choice of local bases.

We employ an alternative, means of identifying quantum channels with bipartite operators. In this approach, basis-dependence is avoided by replacing the reference state with the normalized swap operator $V/d$. Since the swap operator is not a density operator, this correspondence lacks an interpretation as a physical process. But, for our purposes, the lack of physical interpretation comes at a greater benefit. The resulting bipartite operator bears the statistical properties of the corresponding channel.

The identification was introduced for the special case of qubits in \cite{Vedral2013}. We make this idea more precise and general by defining the \emph{homocorrelation map}, $\mathcal{H}$, which takes any map $\map\in\mathcal{L}(\hilbert_A,\hilbert_B)$ (with $\mathcal{L}(\hilbert_A,\hilbert_B)$ being the set of linear maps, or ``superoperators'', from $\mathcal{B}(\hilbert_A)$ to $ \mathcal{B}(\hilbert_B)$), to a ``channel operator'' $M_{H}\in\mathcal{B}(\hilbert_A\otimes\hilbert_B)$ according to
\begin{eqnarray}
\mathcal{H}(\map) \equiv [\mathcal{I}_A\otimes\map] (V/d_A) = \frac{1}{d_A}\sum_{ij} \ket{i}\bra{j} \otimes {\mathcal M}(\ket{j}\bra{i}), 
\label{HomoCor}
\end{eqnarray}
where, again, $V= \sum_{i,j}\ket{ij}\bra{ji}$ with respect to any orthonormal basis $\{\ket{i}\}$.
While the CJ isomorphism is a handy diagnostic tool,
the homocorrelation map serves a different purpose. It \emph{does not} take CP maps to positive operators. Instead, it takes each map to an operator which exhibits the same statistical correlations as that map. This is made precise in the following:
\begin{figure}[t]
\centering
\hspace*{1cm}
\subfigure[]{\includegraphics[width=.35\columnwidth,viewport=50 420 530 750,clip]{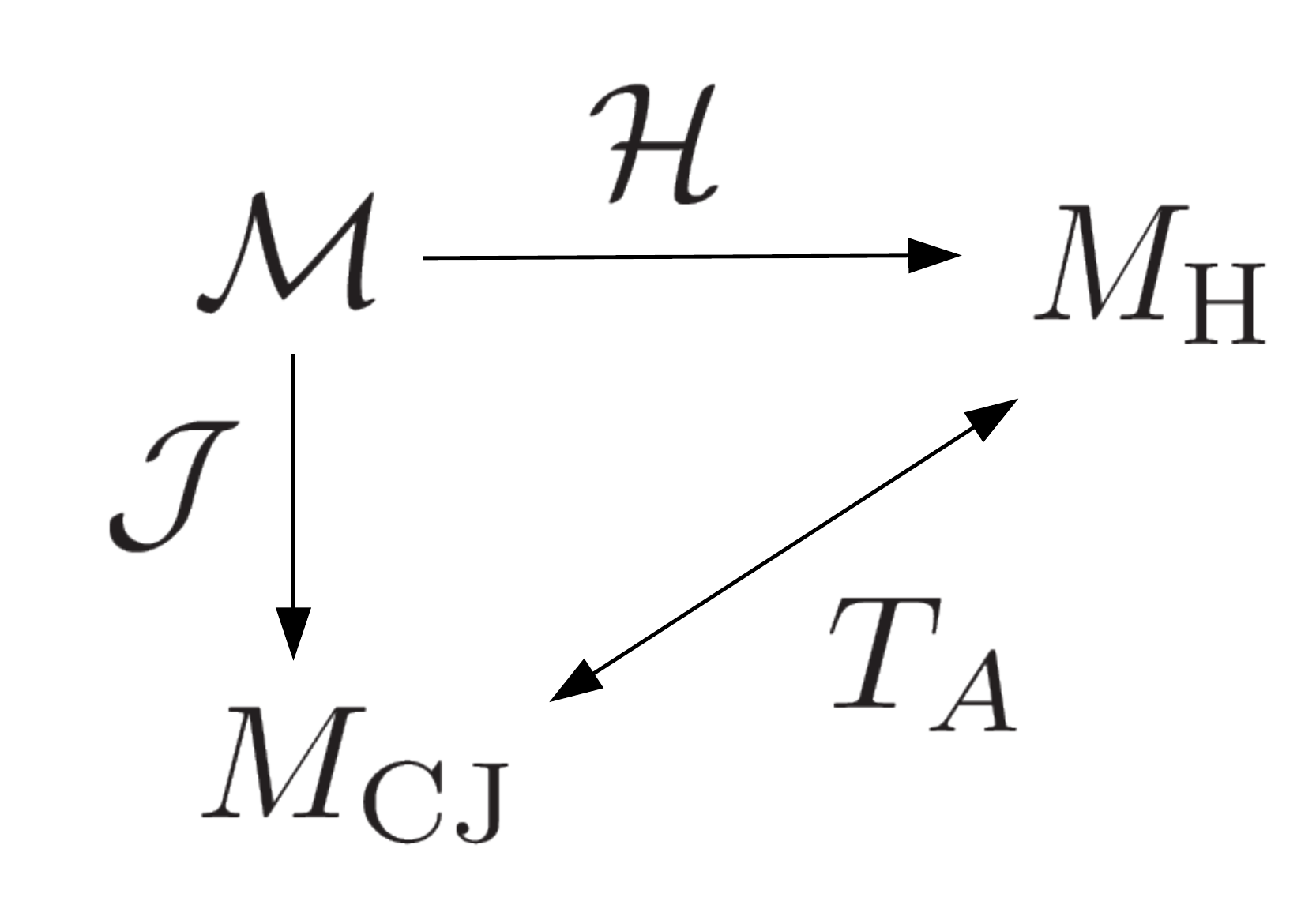}}\hspace*{0.5cm}
\subfigure[]{\includegraphics[width=.35\columnwidth,viewport=50 420 530 750,clip]{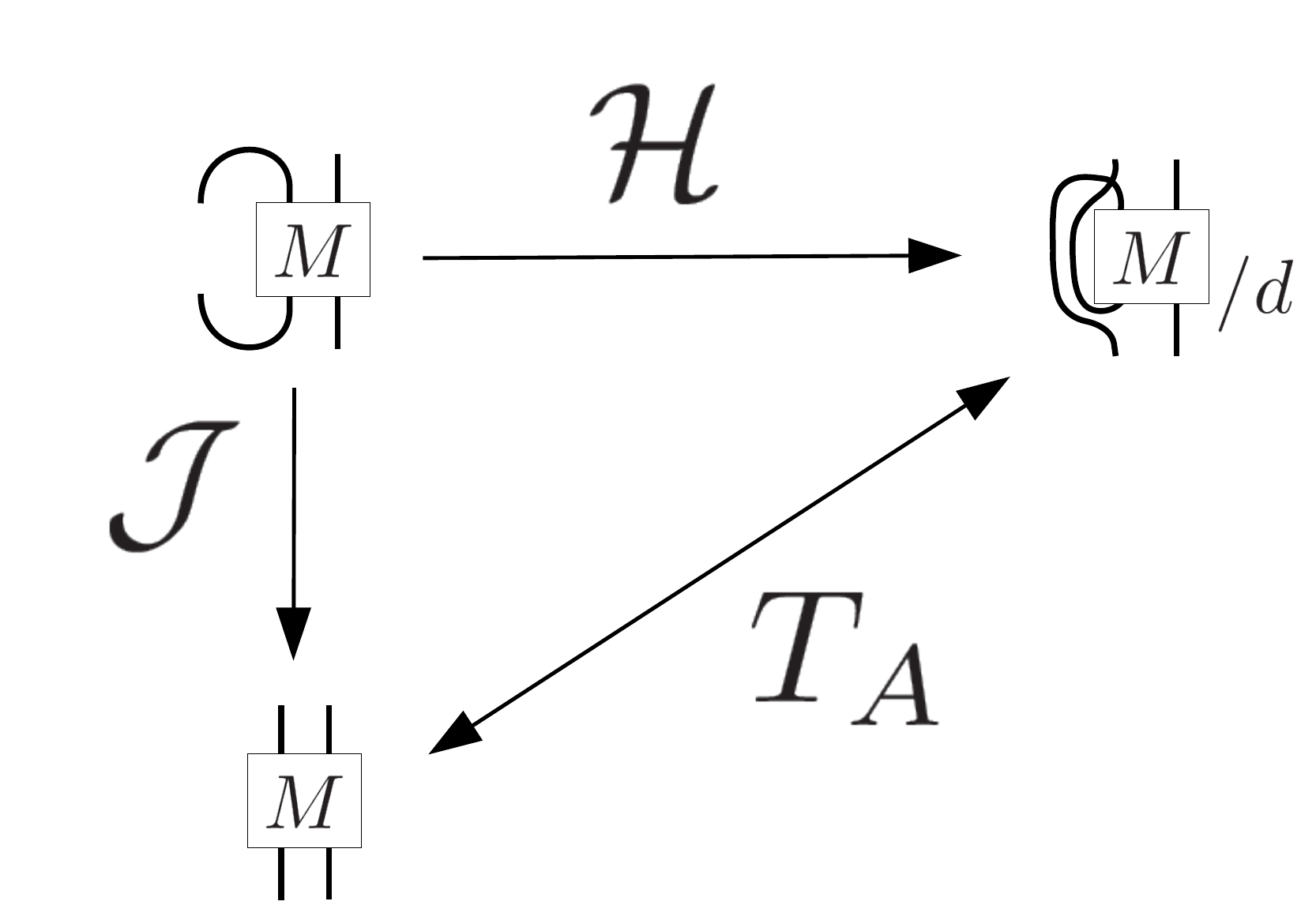}}
\vspace*{-2mm}
\caption{(a) Commutivity diagram summarizing the relationship between the Choi-Jamiolkowski isomorphism and the homocorrelation map defined in Eqs. (\ref{HomoCor})-(\ref{choi}). In (b), the corresponding actions are given in terms of tensor network diagram notation \cite{Biamonte2011}.  Proposition \ref{thm:homocorrprop} may be straightforwardly 
proved using this notation.
}
\label{fig:choihomo}
\end{figure}

\begin{prop}
A bipartite state $\rho\in\mathcal{B}(\hilbert_A\otimes\hilbert_B)$ and a quantum channel $\map:\mathcal{B}(\hilbert_A)\rightarrow\mathcal{B}(\hilbert_B)$ exhibit the same correlations, that is,
\begin{equation}
\label{eq:homocorrprop}
\trn{}[\rho A\otimes B]=\frac{1}{d_A}\trn{}[\map(A)B],\quad \forall A \in \mathcal{B}(\hilbert_A), 
B \in \mathcal{B}(\hilbert_B ). 
\end{equation}
if and only if the equality $\mathcal{H}(\map) = \rho$ holds.
\label{thm:homocorrprop}
\end{prop}
{\em Proof.} The two operators $\rho$ and $\mathcal{H}(\map)$ are equal if and only if their expectations 
$\tr[\rho A\otimes B]=\tr[\mathcal{H}(\map)A\otimes B]$ for all $A,B$. Thus, it suffices to show that 
$\tr[\mathcal{H}(\map)A\otimes B]=\frac{1}{d_A}\tr[\map(A)B]$ for all $A,B$.
This equality may be established as follows:
\begin{eqnarray*}
\tr[\mathcal{H}(\map)A\otimes B]&=\frac{1}{d_A}\sum_{i,j}\tr[\ket{i}\bra{j}\otimes\map(\ket{j}\bra{i})A\otimes B]\nonumber\\
&=\frac{1}{d_A}\sum_{i,j}\tr[(\ket{i}\bra{j}A)\otimes(\map(\ket{j}\bra{i})B)]\nonumber\\
&=\frac{1}{d_A}\sum_{i,j}\bra{j}A\ket{i}\tr[\map(\ket{j}\bra{i})B]=\frac{1}{d_A}\tr[\map(A)B]. \hspace{22mm} \Box
\end{eqnarray*}

Equation (\ref{eq:homocorrprop}) may be taken as the {\em defining} property of the homocorrelation map. 
An example demonstrates the utility of this representation. Consider the one-parameter family of qudit depolarizing channels \cite{Nielsen2001}, defined as 
\begin{equation}
\mathcal{D}_{\eta}(\rho)=(1-\eta)\tr(\rho)\frac{\identity}{d}+\eta\rho.                                                                                                                                                           
\end{equation}
The action of this channel 
commutes with all unitary channels
in that $\mathcal{D}_{\eta}(U\rho U^{\dagger})=U\mathcal{D}_{\eta}(\rho) U^{\dagger}$. 
Under the homocorrelation map, the depolarizing channels are taken to operators with $U\otimes U$ symmetry, namely, 
\begin{equation}
 \mathcal{H}(\mathcal{D}_{\eta})=(1-\eta)\frac{\identity\otimes\identity}{d^2}+\eta \frac{V}{d},
\end{equation}
where $V$ is, again, the swap operator. Trace-one, positive operators of this form are the well-known {\em Werner states} \cite{Werner1989} (see also Sec. \ref{Werner}). Imagine that an observer does not know {\em a priori} whether her two measurements are made on distinct systems in a Werner state or if they are made on the same system before and after a depolarizing channel has been applied. If presented with a Werner state or depolarizing channel having $\eta=-\frac{1}{d^2-1}$ to $\frac{1}{d+1}$, the observer will \emph{not} be able to distinguish between the two cases. 
The homocorrelation map makes this operational identification explicit.
To contrast, the CJ map takes the depolarizing channels to so-called {\em isotropic states} \cite{Horodecki1999},
\begin{equation}
 \mathcal{J}(\mathcal{D}_{\eta})=(1-\eta)\frac{\identity\otimes\identity}{d^2}+\eta\ketbra{\Phi^+},
\end{equation}
where $\ket{\Phi^+}=\sum_i \ket{i i}/\sqrt{d}$. The isotropic states are defined by their symmetry with respect to $U\otimes U^T$ transformations. An observer in the scenario above would certainly be able to distinguish between the correlations of the depolarizing channel and the isotropic states, as long as $\eta\neq0$. 

The distinction between the CJ isomorphism and the homocorrelation map can be further appreciated by contrasting the sets of operators they produce.
The set of CP maps forms a cone in the set of superoperators $\mathcal{L}(\hilbert_A,\hilbert_B)$. Both the CJ isomorphism and the homocorrelation map are cone-preserving maps (by linearity) from $\mathcal{L}(\hilbert_A,\hilbert_B)$ to $\mathcal{B}(\hilbert_A\otimes\hilbert_B)$. While in the case of the CJ isomorphism, the resulting cone \emph{is exactly} the cone of bipartite states, in the case of the homocorrelation map, the cone is distinct from the cone of states. One of the main findings of this paper is that the correlations exhibited by bipartite states and the ones exhibited by quantum channels need not be equivalent. Furthermore, we find that this difference plays a role in their distinct joinability properties. The homocorrelation representation of channels provides us with a natural framework for exploring this difference: a channel and a state with differing correlations will be represented as \emph{distinct} operators in the \emph{same} operator space. These notions and their use in joinability 
are fleshed out in what follows.

The cone of positive operators plays a central role in defining joinability of quantum states. Analogously, the cone of homocorrelation-mapped channels (or ``channel-positive operators'') will play a central role in defining joinability of 
quantum channels. 
\begin{defn} {\bf [State-positivity]}
An operator $M \in \mathcal{B}(\hilbert)$ is \emph{state-positive} if $\trn{}(MP)\geq0$ for all hermitian projectors $P=P^{\dagger}=P^2 \in\mathcal{B}(\hilbert)$. We notate this condition as $M\geq_{st}0$ and emphasize that the resulting set is a self-dual cone.
\end{defn}
\noindent
Recall that a map $\map$ is a valid quantum channel if $\trn{}[\mathcal{J}(\map)P]\geq 0$ for all $P=P^2\in\mathcal{B}(\hilbert_A\otimes\hilbert_B)$ \cite{Choi1975}. Using the relationships of Fig. \ref{fig:choihomo}, we translate this condition to one on the homocorrelation-mapped operator $M=\mathcal{H}(\map)$. 
Specifically, we define:
\begin{defn} {\bf [Channel-positivity]}
An operator $M \in \mathcal{B}(\hilbert_A\otimes\hilbert_B)$ is \emph{channel-positive} with respect to the $A$-$B$ bipartition if $\trn{}(MP^{T_A})\geq0$ for all hermitian projectors $P=P^{\dagger}=P^2\in\mathcal{B}(\hilbert_A\otimes\hilbert_B)$. We notate this condition $M\geq_{ch}0$, and emphasize that the resulting set is, again, a self-dual cone.
\end{defn}

\begin{figure}[h]
\vspace*{-5mm}
\centering
\hspace*{1cm}\includegraphics[width=.35\columnwidth,viewport=120 290 500 670,clip]{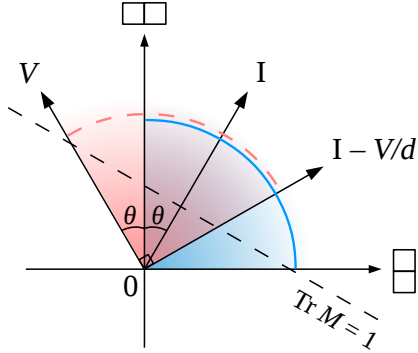}
\vspace*{-2mm}
\caption{State- and channel-positive cones for two qudit Werner operators. The region of the solid arc (blue) corresponds to state-positive operators, while the region of the dashed arc (pink) corresponds to channel-positive operators. The overlapping region, seen as purple, corresponds to PPT operators; of these, the normalized operators are also unentangled state-positive operators. The self-dual nature of the state- and channel-positive cones is consistent with the right angles of each cone's vertex. The Young diagrams represent the corresponding projectors into the symmetric $\frac{1}{2}(\identity +V)$ and antisymmetric $\frac{1}{2}(\identity -V)$ subspaces, respectively. For qudit dimension $d$, the angle $\theta$ is calculated to be $\cos{\theta}=\tr[V(\identity+V)]/\sqrt{\tr[V^2]\tr[(\identity+V)^2]}=\sqrt{({d+1})/{2d}}$.}
\label{fig:statechannelcones}
\end{figure}

In the general case, we can give a characterization of the intersection of the two cones and their complements. This is aided by the fact that the CJ isomorphism and the homocorrelation map are related to one another by partial transpose. A commutivity diagram of these relationships is given in Fig. \ref{fig:choihomo}, where the tensor network diagram calculus \cite{Biamonte2011} may be used to concisely demonstrate that $\mathcal{J}^{-1}\circ\mathcal{H}=\mathcal{H}^{-1}\circ\mathcal{J}=T_A$, up to normalization. 
\begin{prop}
A bipartite state $\rho\in\mathcal{B}(\hilbert_A\otimes\hilbert_B)$ and a quantum channel $\map:\mathcal{B}(\hilbert_A)\rightarrow\mathcal{B}(\hilbert_B)$ exhibit the same correlations if and only if the density operator (or equivalently, channel operator) has a positive partial transpose (PPT).
\label{thm:PPT}
\end{prop}
\begin{proof}
By Prop. \ref{thm:homocorrprop}, if a bipartite state and a quantum channel exhibit the same correlations, then $\rho=\mathcal{H}(\map)$. Since the CJ isomorphism is related to the homocorrelation map by a partial trace, we also have 
$\mathcal{H}(\map)^{T_A}=\mathcal{J}(\map)/d$. By the positivity preservation of the CJ isomorphism, $\map$ being CPTP implies that $\mathcal{J}(\map)$ is a positive operator. Thus, we have that $\rho^{T_A}=\mathcal{J}(\map)/d$ is positive.
\end{proof}
This result may be used to directly connect quantum channels to entanglement:
\begin{coro}
If the correlations of a bipartite state $\rho\in\mathcal{B}(\hilbert_A\otimes\hilbert_B)$ cannot be exhibited by a quantum channel, then the state is entangled.
\end{coro}
\begin{proof}
Since the correlations cannot be exhibited by a quantum channel, the operator is not PPT, by Prop. \ref{thm:PPT}. Then, by the Peres-Horodecki criterion \cite{Peres1996}, the state is necessarily entangled.
\end{proof}
\noindent 
In Section \ref{sec:lp} we will return to the relationship between entanglement, quantum channels, and joinability.

\subsection{Generalization of joinability}
\label{subsec:genjoin}

We are now poised to use the homocorrelation representation to define the joinability of channels. The channel-positive operators provide an alternative set with which to define the allowed joining operators $W$. As a warm-up, we rephrase the channel-joinability problem that was posed in the Introduction. 
Consider quantum channels from $\hilbert_A$ to $\hilbert_B\otimes \hilbert_C$. Under the homocorrelation map, these correspond to tripartite operators lying in the channel-positive cone, notated $W_{A|BC}$. 
The partial traces $\tr_C$ and $\tr_B$ take channel-positive operators in $W_{A|BC}$ to channel-positive operators in $W_{A|B}$ and $W_{A|C}$, respectively; that is, operators in $W_{A|B}$ and $W_{A|C}$ correspond to valid quantum channels via the homocorrelation map.
The corresponding channel-joining scenario is then defined as $(W_{A|BC}, \{\trn{C},\trn{B}\})$.
A channel-joinability problem presents two channel operators $M_{AB}\in W_{A|B}$ and $M_{AC}\in W_{A|C}$ and seeks to determine the existence of a channel operator $M_{ABC}\in W_{A|BC}$ which reduces to the two channel operators in question. In general, we thus have the following:
\begin{defn}
{\bf [Channel-Joinability]} Given a joinability scenario described by the pair 
$(W \geq_{\text{ch}}0\},\{\trn{\ell_k}\})$, the reduced operators 
$\{M_k\}\in\{R_k\}$ are \emph{joinable} if there exists a joint operator $M\in W$ such that $\trn{\ell_k}(M) =M_k$ for all $k$.
\end{defn}
We note that a channel joinability (or extension) problem can be stated using the CJ isomorphism instead of the homocorrelation map, as done in \cite{Chen2012}. However, as we argued, the homocorrelation map provides a platform to 
{\em directly} compare the joinability of states and channels of equivalent correlations. 
For instance, it will allow us to {\em simultaneously} compare the joinability of 
local-unitary-invariant quantum states and channels, and consequently to compare these 
both to the joinability of analogous classical probability distributions (c.f. Fig. \ref{fig:csjoining}).


Before proceeding to the general notion of joinability, 
we also remark that allowed joining operators in $W$ have thus far been considered to be either state-positive or channel-positive. However, from a mathematical standpoint, a sensible joinability problem only needs $W$ to be a {\em convex cone}. To investigate this generalization and (as motivated later) to meld state and channel joining, we consider a third type of positivity that we call {\em local-positivity}. This notion is equivalent to both block-positivity \cite{Jamiolkowski1974} and to map-positivity (not necessarily CP) \cite{Bengtsson2006,Bhatia2009}, in that by representing linear maps using the homocorrelation map, the cone of (transformed) positive maps is equal to the cone of bipartite block-positive operators. 
Formally:
\begin{defn} {\bf [Local-positivity]}
An operator $M \in \mathcal{B}(\hilbert_A\otimes\hilbert_B)$ is \emph{local-positive} with respect to the $A$-$B$ factorization if $\trn{}(MP_A\otimes P_B)\geq0$ for all pure states $P_A=P_A^2\in\mathcal{B}(\hilbert_A)$ and $P_B=P_B^2\in\mathcal{B}(\hilbert_B)$. We notate this condition $M\geq_{loc}0$.
\end{defn}
\noindent 
The set of channel-positive operators and state-positive operators are each subsets (specifically, subcones) of the local-positive operators, as local-positivity clearly is a weaker condition. Local-positive operators are directly relevant to QIP, in particular because they may serve as an entanglement witnesses \cite{Terhal2000}. Moreover, in comparing quantum-joinability limitations to analogous limitations stemming from classical probability theory, joinability scenarios defined with respect to $W\geq_{loc} 0$ may allow the identification of quantum limitations in a ``minimally constrained" setting, closer to the (less strict) classical boundaries. In Sec. \ref{sec:lp}, we find that local-positivity does nevertheless provide stricter-than-classical limitations on joinability.

Another way of viewing the various definitions of positivity is to understand the subscript on the inequality to indicate the dual cone from which inner products with $M$ must be positive. For $M\geq_{\text{st}}0$, $M\geq_{\text{ch}}0$, and $M\geq_{\text{loc}}0$, the respective dual cones are the positive span of rank-one projectors, the positive span of 
partially-transposed projectors, and the positive span of product projectors (from which the trace-one condition confines to the set of separable states). We note that the first two cones are self-dual (and are furthermore, symmetric cones \cite{Jordan1934}), while the local-positive cone is not.
With several important examples of positivity established,
each being a different convex set with which to define $W$, we are in a position to give the following:
\begin{defn} {\bf [General Quantum Joinability]}
\label{defn:genjoin}
Let $W$ be a convex cone in $\mathcal{B}(\hilbert^{(N)})$, and $\{\trn{\ell_k}\}$ be partial traces with $\ell_k\subset\mathbb{Z}_N$. Given the joinability scenario $(W,\{\trn{\ell_k}\})$, the operators $\{M_k\}\in\{R_k\}$ are \emph{joinable} if there exists a joining operator $w\in W$ such that $\trn{\ell_k}(w) =M_k$ for all $k$.
\end{defn}

\noindent
This general definition naturally encompasses the various joinability problems referenced in the Introduction.
Specifically, in the case where $W$ is the set of quantum states on a multipartite system, the joinability problem 
reduces to the quantum marginal problem, while if $W$ consists of channel-positive operators describing quantum channels from one multipartite system to another, one recovers the channel-joining problem instead. Specific instances of this problem are the optimal asymmetric cloning problem \cite{Cerf2000,Iblisdir2005,Kay2012}, the symmetric cloning problem \cite{Werner1998,Zanardi1998}, and the $k$-extendibility problem for quantum maps \cite{Horodecki2013Sep}.
In addition to providing a unified perspective, our approach has the important advantage that different 
classes of joinability problems may be mapped into one another, in such a way that a solution to one provides a solution to another. This is made formal in the following:
\begin{prop}
\label{thm:joinhomomorphism}
Let $W$ and $W'$ be two positive cones 
of operators
acting 
on the space $\hilbert^{(N)}$, let $\{\tr_{\ell_k}\}$ be a set of partial traces that apply to both cones, and let $\phi: W \rightarrow W'$ be a positivity-preserving (homo)morphism, which permits reduced actions $\phi_k$ satisfying  $\phi_k\circ\tr_{\ell_{k}}=\tr_{\ell_{k}}\circ\,\phi$. If $\{M_k\}\in\{\tr_{\ell_k}(W)\}$ is joinable with respect to $W$, then $\{\phi_k(M_k)\}\in\{\tr_{\ell_k}(W')\}$ is joinable with respect to $W'$.
\end{prop}
\begin{proof}
Assume that $w$ is a valid joining operator for the set of operators $\{M_k\}\in\{\tr_{\ell_k}(W)\}$. Then, the set of operators $\{\phi_k(M_k)\}\in\{\tr_{\ell_k}(W')\}$ is joined by the operator $\phi(w)$, since $\tr_{\ell_k}[\phi(w)]=\phi_k(\tr_{\ell_k}[w])=\phi_k(M_k)$ and $\phi(w)\in W'$.
\end{proof}
This is shown in the commutative diagram of Fig. \ref{fig:comm}. We use a stronger corollary of this result in the remaining sections:
\begin{coro}
\label{thm:joinisomorphism}
Let $\phi$ be a one-to-one positivity-preserving map from $W$ to $W'$, with invertible reduced actions $\phi_k$ satisfying $\phi_k\circ\tr_{\ell_{k}}=\tr_{\ell_{k}}\circ\,\phi$ (and similarly for their inverses). Then a set of operators $\{M_k\} \in \{\tr_{\ell_k}(W)\}$ is joinable if and only if the set of operators $\phi_k(M_k)$ is joinable.
\end{coro}
\begin{proof}
The forward implication follows from Proposition \ref{thm:joinhomomorphism}, while the backwards implication follows from the fact that $\phi$ and the $\phi_k$ are invertible, along with the contrapositive of Proposition \ref{thm:joinhomomorphism}.
\end{proof}
The joinability-problem isomorphism we make use of is the partial transpose map. The latter permits a natural reduced action, namely, partial transpose on the remaining of the previously transposed subsystems. As explained, the partial transpose is a positivity-preserving bijection between states and channel operators (via ${\cal H}$). Thus, if we determine the joinable-unjoinable demarcation for a class of states, we will determine the joinable-unjoinable demarcation for a corresponding class of channel-operators. 

\begin{figure}
\centering
\includegraphics[width=.35\columnwidth,viewport=50 410 500 710,clip]{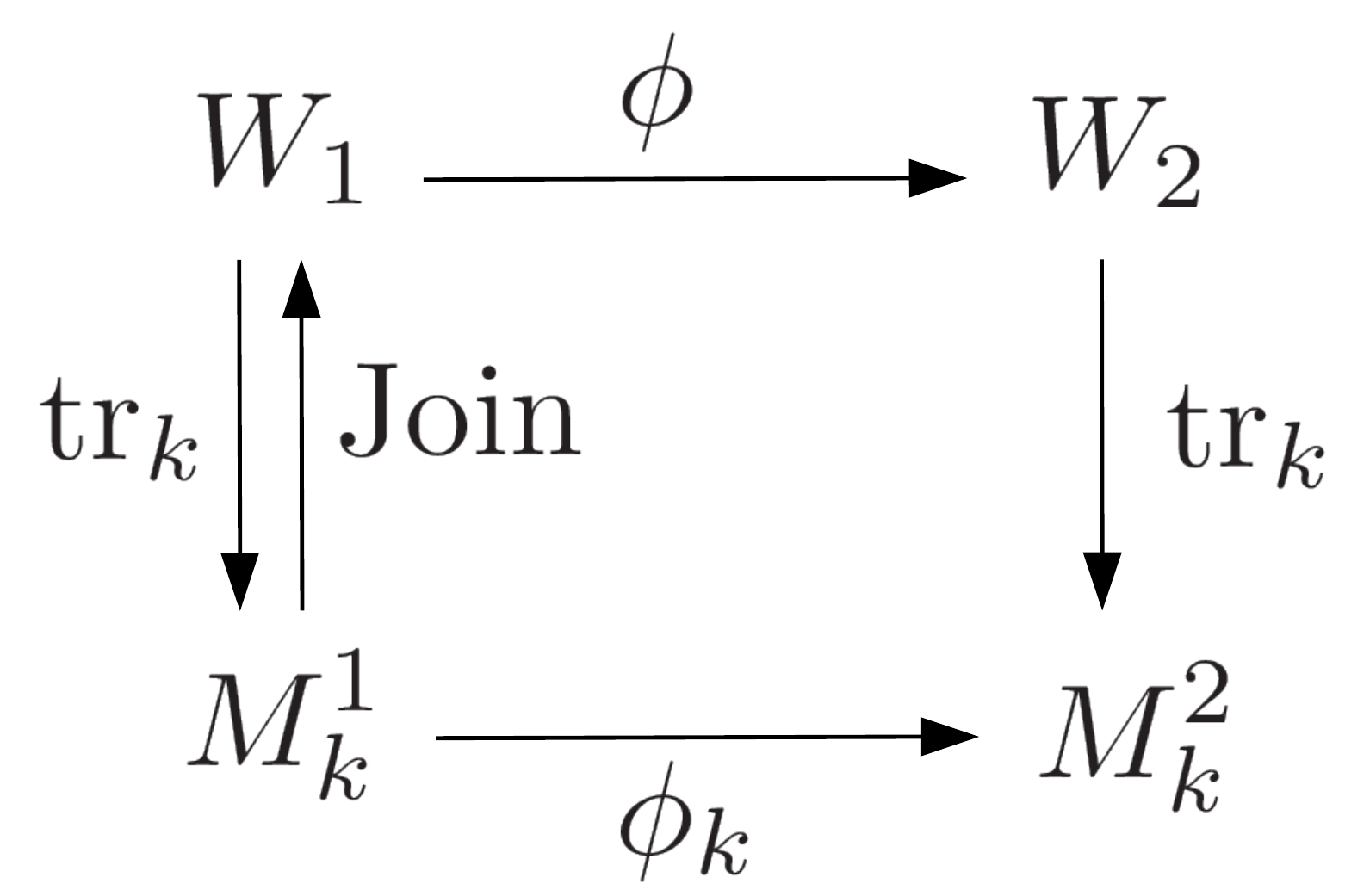}
\vspace*{-3mm}
\caption{Commutivity diagram showing a homomorphism of joinability problems.}
\label{fig:comm}
\end{figure}

\section{Three-party joinability settings with collective invariance}
\label{sec:tripartite}

In this Section, we obtain an exact analytical characterization of the state-joining, channel-joining, and 
local-positive joining problems in the three-party scenario, 
by taking advantage of {\em collective unitary invariance}. That is, we determine
what trio of bipartite operators $M_{AB}$, $M_{AC}$, and $M_{BC}$ may be joined by a valid joining operator 
$w_{ABC}$, subject to the appropriate symmetry constraints. As noted, the most familiar case is state joinability, 
whereby the bipartite operators along with the joining tri-partite operator are state-positive. The next case considered is referred to as ``1-2 channel joinability'': here, we specify a bipartition of the systems (say, $A|BC$) and consider the bipartite operators which cross the bipartition ($M_{AB}$ and $M_{AC}$), along with the joining operator, to be 
channel-positive with respect to the bipartition, while the remaining bipartite operator ($M_{BC}$) is state-positive. Since each of the three possible bipartition choices ($A|BC$, $B|AC$, and $C|AB$) constitutes a different channel joinability scenario, a total of four possibilities arise for three-party state/channel joinability. 
Lastly, motivated by the suggestive symmetry arising from these results and their relation to classical joining, we consider the weaker notion of local-positive joining, in which all operators involved are only required to be local-positive.

\subsection{Joinability limitations from state-positivity and channel-positivity}
\label{Werner}

We begin by describing the operators which are to be joined. The bipartite reduced operators inherit the collective unitary invariance from the tripartite operators from which they are obtained. Therefore, by a standard result of representation theory \cite{Weyl1997}, 
the operators which are to be joined are of the following form:
\begin{equation}\label{eq:Werner}
\rho(\eta)=(1-\eta)\frac{\identity}{d^2}+\eta \frac{V}{d},
\end{equation}
where $V$ is the swap operator defined earlier.
The above operators are known to be state-positive for the range $-\frac{1}{d-1}\leq\eta \leq\frac{1}{d+1}$, 
corresponding to the $d$-dimensional (qudit) Werner states we already mentioned. The parameterization is chosen so that 
$\eta$ is a ``correlation'' measure: if  $d=2$, $\eta=-1$ corresponds to the singlet state, $\eta=0$ to the maximally mixed state, 
while $\eta=1$ is not a valid quantum state, but expresses perfect correlation for all possible collective measurements. 
Note that a value $\eta=1$, for instance, does correspond to a valid quantum channel. Intuitively, channel-positive operators with 
$U\otimes U$-invariance correspond to depolarizing channels. It is known that complete positivity (or channel-positivity) of the depolarizing channel is ensured provided that $-\frac{1}{d^2-1}\leq\eta \leq1$ \cite{King2003}. However, we find it instructive to independently establish state- and channel- positivity bounds using the CJ isomorphism. 

To this end, we enlarge the above class of $U\otimes U$-invariant operators to the class of operators with {\em collective orthogonal invariance}, namely, invariance under transformations of the more general form $O\otimes O$, belonging to the so-called {\em Brauer algebra} \cite{Brauer1937, Nazarov1996}\footnote{Operators in this algebra have been extensively analyzed in \cite{Werner2002,Keyl2002}, and recent work characterizing their irreducible representations may be found in \cite{Horodecki2013May, Horodecki2013Aug}. The Brauer algebra acting on $N$ $d$-dimensional Hilbert spaces is spanned by representations of subsystem permutations $\{V_{\pi}|\pi\in S_N\}$, along with their partial transpositions with respect to groupings of the subsystems $\{V_{\pi}^{T_l}|\pi\in S_N, l\subseteq\{1,\ldots,N\}\}$. In terms of tensor network diagrams, each element of this basis is represented by a set of disjoint pairings of $2N$ vertices, with the vertices arranged in two rows, both containing $N$ of them.}.
In addition to $U\otimes U$-invariant operators, the Brauer algebra also contains 
$U^*\otimes U$-invariant operators. The latter class of operators, which includes the well-known isotropic states, 
are spanned by the operators $\identity$ and $ V^{T_A}$. Thus, the set of $O\otimes O$-invariant operators are of the form
\begin{equation}\label{eq:Brauer}
\rho(\eta, \beta)=(1-\eta-\beta)\frac{\identity}{d^2}+\eta \frac{V}{d}+\beta \frac{V^{T_A}}{d}.
\end{equation}
In particular, the operator $\rho(0,1)$ is a generic Bell state on two qudits, 
$\rho(0,-1/(d-1))$ is the maximally entangled Werner state (namely, the singlet state for $d=2$), $\rho(1,0)$ is the identity channel, and $\rho(0,0)$ is the completely mixed state (or the completely depolarizing channel). We can then establish the following:

\begin{prop}
\label{prop1}
A bipartite operator $\rho(\eta)$ with collective unitary invariance is channel-positive 
if and only if $-\frac{1}{d^2-1}\leq\eta \leq1$.
\end{prop} 
\begin{proof}
The Brauer algebra includes all the state-positive operators which are mapped, via the CJ isomorphism ${\cal J}$, to the $U\otimes U$-invariant channel-positive operators; under ${\cal J}$, $\eta$ and $\beta$ in Eq. (\ref{eq:Brauer}) are swapped with one another. Since ${\cal J}$ takes state-positive operators to channel-positive operators, we need only obtain the set of state-positive operators. State-positivity of these operators is enforced by the inner products with respect to their (operator) eigenspaces being non-negative. Such eigenspaces are $P_A$, $P_+$, and $P_Y$, independent of $\eta$ and $\beta$: the first is the anti-symmetric subspace, the second is the one-dimensional space spanned by $\ket{\Phi^+} \equiv \sqrt{1/d} \sum_i \ket{ii}$, and the third is the space spanned by vectors $\ket{y}$ satisfying $\bra{y}(\ket{y})^*=0$\footnote{Both the definition of $\ket{\Phi^+}$ and the use of complex conjugation are {\em basis-dependent} notions. It is understood 
that all usages of either refer to the same (arbitrary) choice of basis.}. The eigenvalues are as follows:
\begin{eqnarray*}
\rho(\eta,\beta)P_A &=[(1-\eta-\beta)/d^2-\eta/d]P_A ,\\
\rho(\eta,\beta)P_+ &=[(1-\eta-\beta)/d^2+\eta/d+\beta]P_+,\\
\rho(\eta,\beta)P_Y &=[(1-\eta-\beta)/d^2+\eta/d]P_Y.
\end{eqnarray*}
Hence, state-positivity of the bipartite Brauer operators is ensured by
\begin{eqnarray*}
1  \geq (d+1)\eta+\beta, \quad 
1  \geq -(d-1)\eta-(d^2-1)\beta, \quad 
1  \geq -(d-1)\eta+\beta.
\end{eqnarray*}
The inequalities bounding channel-positivity are obtained by swapping the $\eta$s 
and $\beta$s.
In particular, we recover that the state-positive range for  $U\otimes U$-invariant operators is 
$-\frac{1}{d-1} \leq \eta \leq \frac{1}{d+1}$, whereas the channel-positive range  is $-\frac{1}{d^2-1}\leq\eta\leq1$.
\end{proof}

In a similar manner, we can also obtain the ranges of local-positivity:
\begin{prop}
\label{prop2}
A bipartite operator $\rho(\eta)$ with collective unitary invariance is local-positive 
if and only if $-\frac{1}{d-1}\leq\eta\leq 1$.
\end{prop} 
\begin{proof}
Local positivity is ensured by the non-negativity of expectation values with respect to the product vectors $\{\ket{x x},\ket{x \bar{x}},\ket{y y},\ket{y \bar{y}}\}$, satisfying $\ket{x}^*=\ket{x}$ and $\ket{y}^*=\ket{\bar{y}}$, where the bar indicates a vector orthogonal to the original vector. In terms of $\eta$ and $\beta$, these constraints read \begin{eqnarray*}
0\leq \exval{\rho(\eta,\beta)}{xx}&=(1-\eta-\beta)/d^2+\eta/d+\beta/d,\\
0\leq \exval{\rho(\eta,\beta)}{x\bar{x}}&=(1-\eta-\beta)/d^2,\\
0\leq \exval{\rho(\eta,\beta)}{yy}&=(1-\eta-\beta)/d^2+\eta/d,\\
0\leq \exval{\rho(\eta,\beta)}{y\bar{y}}&=(1-\eta-\beta)/d^2+\beta/d.
\end{eqnarray*}
More compactly, these boundaries are given by
\begin{eqnarray*}
-\frac{1}{d-1}\leq \eta +\beta \leq 1, \quad
 -(d-1)\eta+\beta \leq 1, \quad 
\eta -(d-1)\beta \leq 1.
\end{eqnarray*}
Thus, for bipartite Brauer operators, local-positive operators are equivalent to convex combinations of state- and channel-positive operators\footnote{This property is known as \emph{decomposability} \cite{Bengtsson2006}. Interestingly, such an equivalence also holds for {\em arbitrary} bipartite qubit states \cite{Woronowicz1976}.}. In particular, for the local-positive range of $U\otimes U$-invariant operators, it follows that $-\frac{1}{d-1}\leq\eta\leq 1$, as stated.
\end{proof}

\begin{figure}[t]
\centering
\subfigure[]{
\hspace*{1cm}\includegraphics[width=.25\columnwidth,viewport=50 180 410 540,clip]{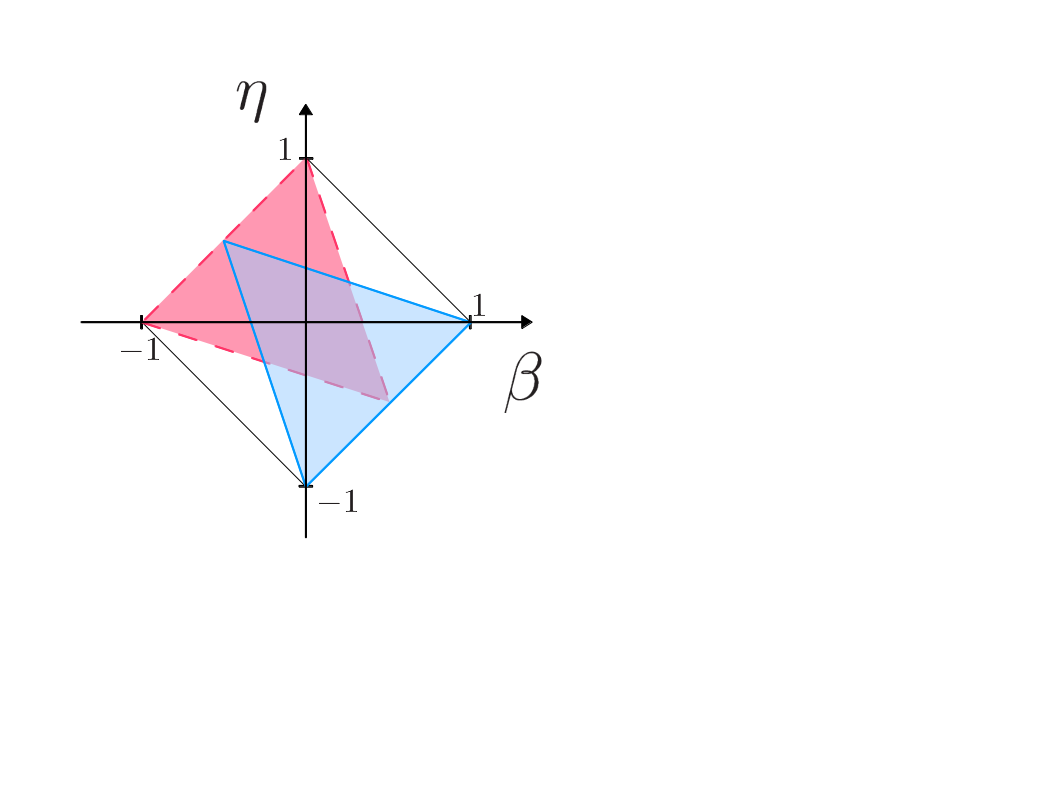}}
\hspace*{1cm}
\subfigure[]{
\includegraphics[width=.24\columnwidth,viewport=170 280 400 530,clip]{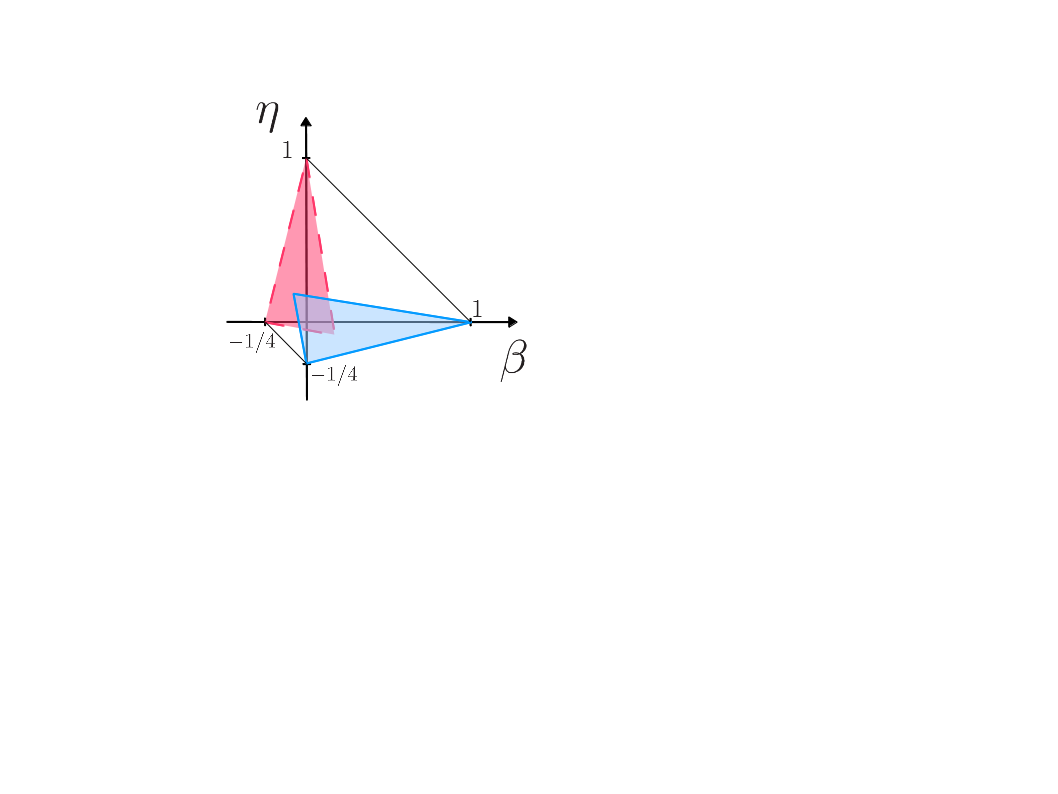}}
\vspace*{-2mm}
\caption{Positivity regions for bipartite Brauer operators: (a) $d=2$. (b) $d=5$. The solid triangle (blue) encloses the state-positive region, the dashed triangle (pink) encloses the the channel-positive region, and the outer boundary encloses the local-positive region. }
\label{fig:brauersummary}
\end{figure}
\noindent 
A pictorial summary of the three positivity bounds 
is presented in Fig. \ref{fig:brauersummary}. 

Having characterized all types of positivity for the (bipartite) operators to be joined, we now turn to characterize the positivity for the (tripartite) joining operators. For each positive tripartite set ($W\geq_{st}0$, $W\geq_{ch}0$, and $W\geq_{loc}0$), we obtain the trios of joinable bipartite operators by simply applying the three partial traces $(\Tr_{A},\Tr_{B},\Tr_{C})$ to each positive operator. In more detail, our approach is to obtain an expression for the positivity boundary of the tripartite operators in terms of operator space coordinates, and then re-express this boundary in terms of reduced-state parameters (the three Werner parameters in this case). For state- and channel-positivity, the desired characterization follows directly from the analysis reported in our previous work \cite{Johnson2013}. 

\begin{coro} With reference to the parameterization of Eq. (\ref{eq:Werner}), we have that: \\
(i) Three Werner states with parameters $(\eta_{AB}, \eta_{AC}, \eta_{BC})$ are joinable with respect to 
the $(W_{ABC}\geq_{\text{st}}0, \{\trn{A},\trn{B},\trn{C}\})$ scenario if and only if 
\begin{eqnarray*}
\left \{ \begin{array}{l}\frac{1}{2}(1-\eta_{AB}-\eta_{AC}-\eta_{BC})\geq 
|\eta_{AB}+\omega\eta_{AC}+\omega^2 \eta_{BC}|, \quad \omega\equiv e^{i 2\pi/3}, \\
\eta_{AB}+\eta_{AC}+\eta_{BC}\geq -1, \end{array}\right. 
\end{eqnarray*}
for $d=2$, while for $d\geq3$ they need only satisfy
\begin{eqnarray*}
\frac{3}{2d(d\mp1)} \Big(1 \pm (d \mp 1) (\eta_{AB}+\eta_{AC}+\eta_{BC}) \Big) \geq 
|\eta_{AB}+\omega\eta_{AC}+\omega^2\eta_{BC}|.
\end{eqnarray*}
\noindent 
(ii) Three $U\otimes U$-invariant operators with parameters $(\eta_{AB}, \eta_{AC}, \eta_{BC})$ 
are channel-joinable with respect to the $(W_{A|BC}\geq_{\text{ch}}0, \{\trn{A},\trn{B},\trn{C}\})$ scenario if and only if 
\begin{eqnarray*}
\hspace*{-3mm}
\frac{1}{d-1}+\eta_{AB}+\eta_{AC}-\eta_{BC}
\geq \bigg|\frac{2}{d-1}+d\eta_{BC}+\sqrt{\frac{2d}{d-1}}(e^{i\theta}\eta_{AB}+e^{-i\theta}\eta_{AC})\bigg|,\\
\hspace*{-3mm}e^{i\theta} \equiv \sqrt{(d-1)/2d}\pm i \sqrt{(d+1)/2d}.
\end{eqnarray*}
The channel-joinability limitations in the other two scenarios $B|AC$ and $C|AB$ may be obtained by 
permuting the $\eta$s accordingly.
\end{coro}
\begin{proof} Result (i) corresponds to Theorem 3 in \cite{Johnson2013}, 
re-expressed in terms of the parametrization of Eq. (\ref{eq:Werner}) (with reference to the notation of 
Eqs. (15)--(17) in \cite{Johnson2013}, one has $\eta_\ell = (d/(d^2-1))(\Psi_\ell^- -1/2)$, $\ell=AB, AC, BC$).
 
In order to establish (ii), note that ${\cal J}$ may be used to translate any $U^{*}\otimes U$-invariant state-positive joinability problem into a $U\otimes U$-invariant channel-positive joinability problem, drawing on Corollary  \ref{thm:joinisomorphism}. 
Explicitly, under ${\cal J}$ (partial transpose in the case of operators), the $U^{*}\otimes U\otimes U$-invariant state-positive operators $W_{A^*BC}$ are in one-to-one correspondence with the $U\otimes U\otimes U$-invariant channel-positive operators $W_{A|BC}$. Hence, by the joinability isomorphism induced by the partial transpose, the solution to a joinability problem of the scenario $(W_{A^{*}BC},\{\trn{A},\trn{B},\trn{C}\})$ gives a solution to a corresponding joinability problem of $(W_{A|BC},\{\trn{A},\trn{B},\trn{C}\})$.
Thus, to obtain the depolarizing channel-joinability boundaries, we simply translate the isotropic state parameters of Eqs. (20)-(21) in \cite{Johnson2013} into $\eta$ parameters.
\end{proof}

\begin{figure}[ht]
\centering
\subfigure[]{
\hspace*{20mm}\includegraphics[width=.4\columnwidth,viewport=110 520 350 720,clip]{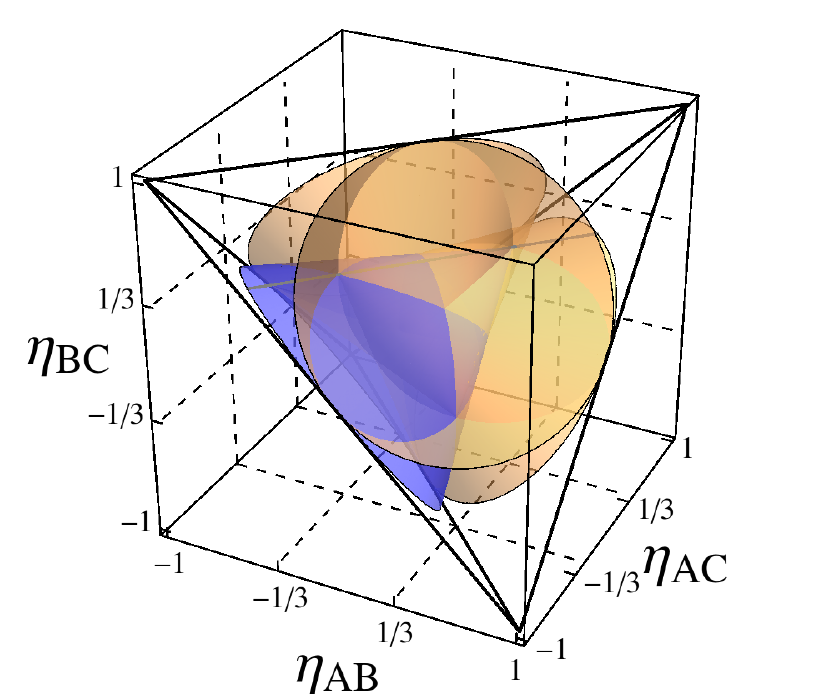}}
\subfigure[]{
\includegraphics[width=.4\columnwidth,viewport=110 500 330 720,clip]{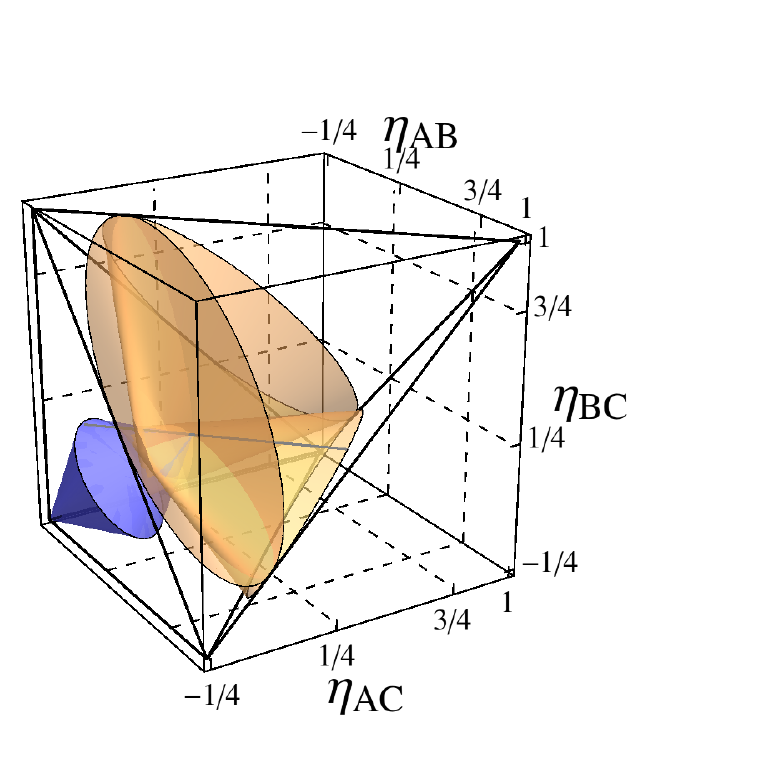}}
\vspace*{-2mm}
\caption{Joinability of operators on the $\beta=0$ line for (a) $d=2$ and (b) $d=5$ (each from a different perspective). State-positivity, along with channel-positivity with respect to each of the three bipartitions, obtains the four cones depicted here. The joinability limitations for classical probability distributions are given by (a) the tetrahedron with black edges and (b) the union of the two tetrahedra with black edges.}
\label{fig:csjoining}
\end{figure}

The joinability limitations of all four scenarios are depicted in Fig. \ref{fig:csjoining}. As stressed in \cite{Johnson2013}, the quantum joinability limitations must adhere to the analogous classical joinability limitations (seen as the tetrahedra in Fig. \ref{fig:csjoining}). 
In the qubit case, we find it intriguing that the inclusion of the quantum channel-joinability limitations allows us to regain the {\em tetrahedral symmetry} imposed by the classical limitations; whereas each scenario on its own expresses a continuous rotational symmetry that is {\em not} reflected classically. In other words, if we consider the joinability scenario defined by 
$$(\text{span}\{W_{ABC}, W_{A|BC}, W_{B|AC}, W_{C|AB}\},\{\trn{A},\trn{B},\trn{C}\}),$$ 
the joinable bipartite operators respect the tetrahedral symmetry suggested by the classical joinability bounds. This amounts to asking the question: what trios of bipartite correlations -- as derivable from {\em either} quantum states or channels, {\em or} from probabilistic combinations of the two -- may be obtained from the measurements on three systems? Though the result expresses the tetrahedral symmetry of the classical joinability limitations, these classical joinability limitations {\em do not suffice} to enforce the stricter quantum joinability limitations, as manifest in the fact that the corners of the classical joinability tetrahedron are {\em not} reached by the quantum boundaries. We diagnose such limitations as strictly quantum features that do not have classical analogues -- as we will discuss later in this work.

\subsection{Joinability limitations from local-positivity}
\label{sec:lp}

We now explore how local-positive joinability (a strictly \emph{weaker} restriction, as noted) 
relates to the state/channel-joinability limitations above, as well as to the underlying classical limitations. 
As of yet, we only know that the local-positive limitations will lie between the classical and the quantum boundaries.
Since obtaining a simple analytical characterization for arbitrary dimension $d$ appears challenging in the 
local-positive setting, and useful insight may already be gained in the lowest-dimensional (qubit) setting, we focus 
on $d=2$ in this section. Our main result is the following:

\begin{thm}
With reference to Eq. (\ref{eq:Werner}), 
three qubit Werner operators (constrained by local-positivity) with parameters $(\eta_{AB}, \eta_{AC}, \eta_{BC})$ are joinable by 
a local-positive tripartite Werner operator $w$ if an only if the following conditions hold:
\begin{eqnarray*}
1+\eta_{AB}+\eta_{AC}+\eta_{BC}\geq & 0, \qquad 1+ \eta_{AB}-\eta_{AC}-\eta_{BC}\geq & 0 , \\
1-\eta_{AB}+\eta_{AC}-\eta_{BC}\geq & 0, \qquad 1-\eta_{AB}-\eta_{AC}+\eta_{BC}\geq  & 0, 
\end{eqnarray*}
and
\begin{eqnarray*}
\hspace*{-1cm}
2\eta_{AB}\eta_{AC}\eta_{BC}-\eta^2_{AB}\eta_{AC}^2-\eta^2_{AB}\eta_{BC}^2-\eta^2_{AC}\eta_{BC}^2 \geq 0.
\end{eqnarray*}
\label{thm:locpos}
\end{thm}


\begin{figure}[t]
\centering
\includegraphics[width=.4\columnwidth,viewport=110 490 350 730,clip]{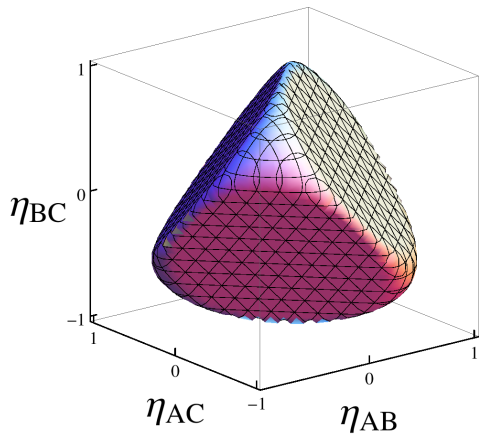}
\vspace*{-2mm}
\caption{Boundary of local-positive Werner operators which are joinable via local-positive operators,
as described by the Roman surface, see Theorem \ref{thm:locpos}.}
\label{fig:locjoininga}
\end{figure}

The proof is rather lengthy and deferred to a separate Appendix. The resulting boundary is depicted in Fig. \ref{fig:locjoininga}; the shape and its determining equation is recognized as the convex hull of the \emph{Roman surface} (aka {\em Steiner surface})
\cite{Bengtsson2006,Bengtsson2011}. Comparing with Fig. \ref{fig:csjoining}(a), we see that, still, the quantum joinability limitation arising from from local-positivity is {\em stricter} than the corresponding classical one. However, it is closer to the classical limitations than the state/channel-positive limitations obtained in the previous section for $d=2$. To shed light on the cause of the quantum boundary here, we can explicitly construct a product-state projector, whose probability would be negative {\em if} joinability outside of this shape were allowed. The family of joining states $w$ that we need to consider (see Appendix) may be parameterized in terms of the bipartite reduced state Werner parameters as
$$
w(\eta_{AB},\eta_{AC},\eta_{BC})=\frac{1}{8}\identity+\frac{\eta_{AB}}{4}(V_{AB}-\identity/2)+\frac{\eta_{AC}}{4}(V_{AC}-\identity/2)+\frac{\eta_{BC}}{4}(V_{BC}-\identity/2).
$$
Consider the following 
state on $A$-$B$-$C$:
\begin{equation}
 \ket{\psi}=\left[
\begin{array}{c}
1 \\ 0
\end{array} \right]\otimes \left[ \begin{array}{c}
\cos{2\pi/3} \\
\sin{2\pi/3} \end{array} \right]\otimes \left[
\begin{array}{c}
\cos{4\pi/3} \\
\sin{4\pi/3} \end{array} \right],
\end{equation}
which corresponds to the pure product state with the local Bloch vectors as anti-parallel with one another as possible. Computing its expectation with respect to 
$w(\eta_{AB}=\eta_{AC}=\eta_{BC}\equiv \eta)$, the largest value of $\eta$ that admits a non-negative value is 
$\eta=2/3$. 
Hence, local-positivity limits the simultaneous joining of these Werner operators to a maximum of $\eta=2/3$. The operational interpretation of this result deserves attention. Consider a local projective measurement made on each of three qubit systems. Furthermore, consider the three systems to have a collective unitary symmetry, in the sense that there are no preferred local bases. In our general picture, where local positivity is considered, the systems need not be three distinct systems -- they may also be the same system at two different points in time.
Local positivity enforces the rule that ``all probabilities arising from such measurements must be non-negative''. In the example above (i.e. $\eta_{AB}=\eta_{AC}=\eta_{BC}$), this rule implies that the three equal correlations (as measured by the $\eta$s) can never exceed 2/3.
As this example and Fig. \ref{fig:locjoininga} show, local-positivity enforces joinability limitations {\em more strict} than those of classical probability theory. 
Notwithstanding, 
the limitations arising from local-positivity reflect the same symmetry as the classical limitations do, namely, symmetry with respect to individually inverting two axes. The state-joining and channel-joining scenarios reflected a preference towards the negative axis (anticorrelation) and the positive axis (correlation) of the $\eta$s, respectively.

Before concluding this section, we connect the above discussion to the relationship between local-positivity and separability. As mentioned earlier, the cone of local positive operators and the cone of separable operators are dual to one another. The operator subspace we are dealing with is spanned by the orthonormal operators $\frac{1}{\sqrt{8}}\identity$, $\frac{1}{\sqrt{6}}(V_{AB}-\identity/2)$, $\frac{1}{\sqrt{6}}(V_{AC}-\identity/2)$, and $\frac{1}{\sqrt{6}}(V_{BC}-\identity/2)$ with coordinates $\frac{1}{\sqrt{8}}$, $\sqrt{\frac{3}{8}}\eta_{AB}$, $\sqrt{\frac{3}{8}}\eta_{AC}$, and $\sqrt{\frac{3}{8}}\eta_{BC}$, respectively.
In Theorem \ref{thm:locpos}, we determined the algebraic surface bounding the local positive operators; hence, the dual to this surface will bound the separable operators within this space. The dual to the Roman surface is known as the {\em Cayley's cubic surface} \cite{Henrion2011}, which, for a given scale parameter $w$ is characterized by
$$
\left| 
\begin{array}{ccc}
w & x & y \\
x & w & z \\
y & z & w
\end{array} \right| = 0.$$
We first set $x=\sqrt{\frac{3}{8}}\eta_{AB}$, $y=\sqrt{\frac{3}{8}}\eta_{AC}$, and $z=\sqrt{\frac{3}{8}}\eta_{BC}$. Then we must set $w$ so that the Cayley surface delimits the separable states. For each extremal separable state in our space, there is a corresponding local-positive operator acting as an entanglement witness; a state is separable if the inner product with its entanglement witness is nonnegative.

Consider the extremal local-positive operator $\eta_AB=\eta_AC=\eta_BC=2/3$ that we made use of previously. This operator will act as an entanglement witness for another operator with $\eta_AB=\eta_AC=\eta_BC=\sigma$. We obtain $\sigma$ by solving 
$$
\left[ 
\begin{array}{c}
\frac{1}{\sqrt{8}} \\
\sqrt{\frac{3}{8}}\frac{2}{3} \\
\sqrt{\frac{3}{8}}\frac{2}{3} \\
\sqrt{\frac{3}{8}}\frac{2}{3} 
\end{array} \right]
\cdot\left[ 
\begin{array}{c}
\frac{1}{\sqrt{8}} \\
\sqrt{\frac{3}{8}}\sigma \\
\sqrt{\frac{3}{8}}\sigma\\
\sqrt{\frac{3}{8}}\sigma 
\end{array} \right]=0,
$$
to arrive at $\sigma=-\frac{1}{6}$. With this, the only value of $w$ allowing the Cayley surface to be solved by $\sigma=-\frac{1}{6}$ is $w=\frac{1}{\sqrt{24}}$. Setting the scaling value and evaluating the determinant, we find that the separable states are bound by the surface
\begin{eqnarray}
1+54\eta_{AB}\eta_{AC}\eta_{BC}-9(\eta_{AB}+\eta_{AC}+\eta_{BC})^2 \nonumber \\ 
+18(\eta_{AB}\eta_{AC}+\eta_{AB}\eta_{BC}+\eta_{AC}\eta_{BC})\geq 0.
\end{eqnarray}
This inequality may also be obtained using Theorem 1 in \cite{Eggeling2001}. The shape of the separable states is depicted in Figure \ref{fig:locposandseparable}. 
Several remarks may be made. First, the set of separable states exhibits the tetrahedral symmetry shared by the classical joinability boundary and the local-positive joinability boundary. Thus, among the various boundaries we have considered in this three dimensional Euclidean space, the state- and channel-positive boundaries are the only ones not obeying tetrahedral symmetry. However, both the convex hull \emph{and} the intersection of the state- and channel-positive cones bound regions which recover this symmetry. It is a curious observation that the convex hull of these cones is ``nearly'' the local-positive region, while the intersection is ``nearly'' the set of separable states. Earlier we found, in the two-qudit case, that local-positivity coincides with the union of the state- and channel-positive regions, as well as that the separable region was their intersection. Here we consider the analog for three qubits. The result is that i) the convex hull of state- and channel- positive operators is strictly 
contained in the set of local-positive operators; and ii) the intersection of the state- and channel-positive operators is strictly contained in the set of separable states.

We may further interpret the latter result in terms of PPT considerations. The operators which result from a homocorrelation-mapped channel necessarily have PPT. Corollary 1 in Ref. \cite{Eggeling2001} states that the PPT and bi-separable Werner operators coincide. Thus, any state-positive operator which is also a homocorrelation mapped channel is necessarily bi-separable. Hence, the intersection of the four cones will be the set of states which are bi-separable with respect to any of the three partitions. This set is clearly contained in the set of tri-separable states.
These observations illuminate the relationships among entanglement, quantum states, and quantum channels. Specifically, the homocorrelation map allows us to place quantum channels in the same arena as quantum states, and hence to directly compare and contrast them. Finding that the tri-separable operators are a proper subset of the bi-separable ones, we wonder what features these strictly bi-separable operators possess, and what does 
bi-separability imply for the states \emph{or} channels 
supporting 
such correlations.

\begin{figure}[t]
\centering

\centering
\subfigure[]{
\hspace*{15mm}
\includegraphics[width=.36\columnwidth,viewport=110 500 350 720,clip]{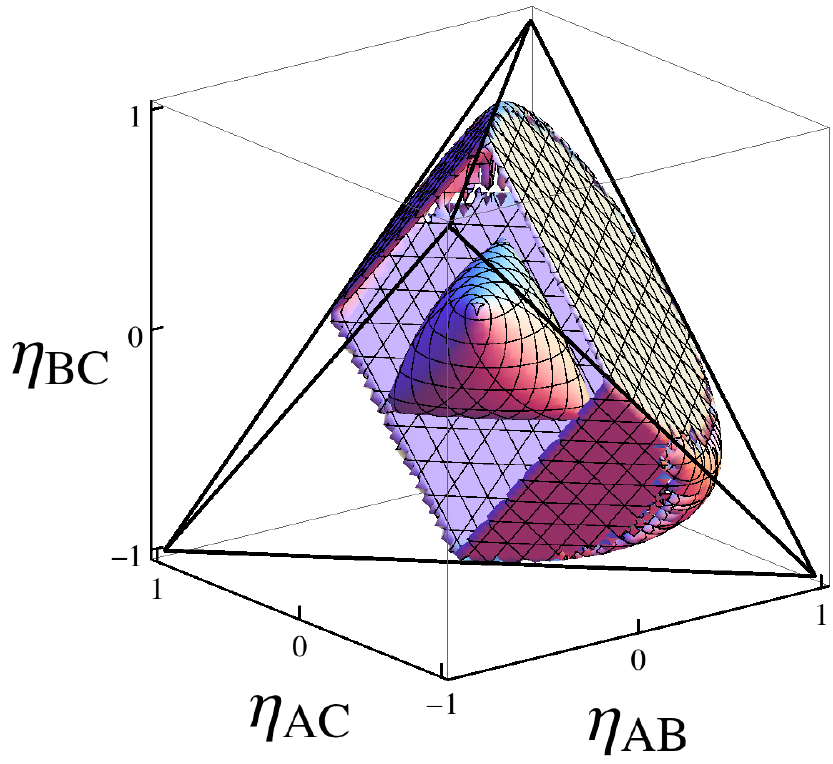}
}
\subfigure[]{
\includegraphics[width=.36\columnwidth,viewport=110 500 350 720,clip]{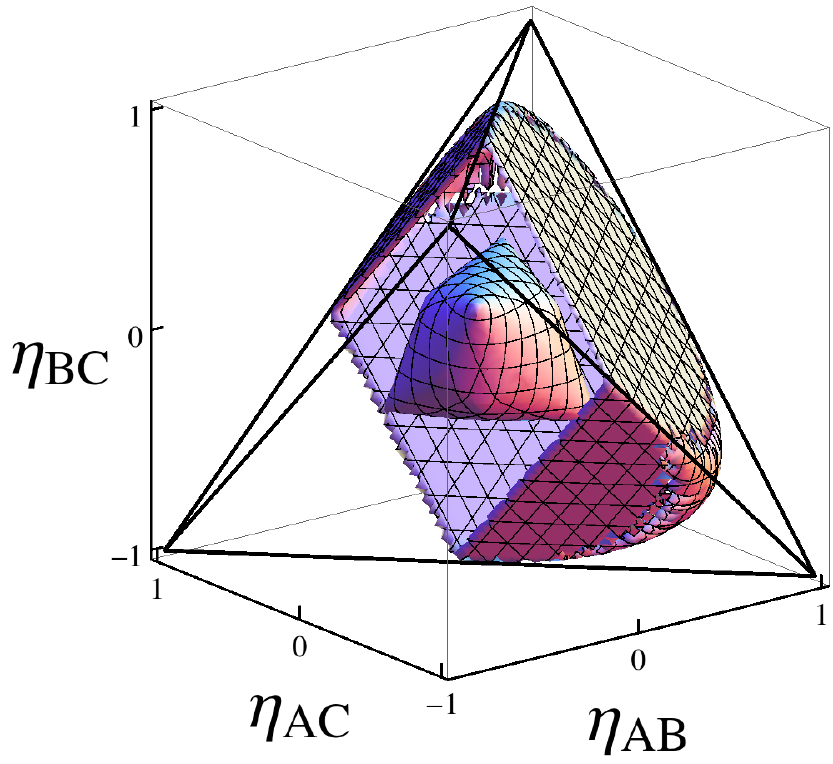}
}
\vspace*{-2mm}
\caption{(a) Set of separable operators within the set of local-positive operators. (b) 
Intersection of the state- and channel-positive cones within the set of local-positive operators. While the separable operators are a subset of the intersection set, the two objects coincide (only) at their vertices. In panel (a) the closest point in the separable set to the $(-1, -1, -1)$ corner of the figure is $(-1/6,-1/6,-1/6)$, whereas in panel (b), the closest point in the intersection set is $(-1/5, -1/5, -1/5)$.}
\label{fig:locposandseparable}
\end{figure}

\section{Agreement bounds for quantum states and channels}
\label{sec:agreement}

In what remains, we illustrate some crucial differences between channel- and state-positive operators. These differences inform the nature of their respective joinability limitations. In order to directly compare states to channels we restrict our considerations here to operators in $\mathcal{B}(\hilbert_d\otimes \hilbert_d)$. Qualitatively, state-positive operators are restricted in the degree to which they can support \emph{agreeing} outcomes, whereas channel-positive operators are restricted in the degree to which they can support \emph{disagreeing} outcomes. We define the degree of agreement to be the likelihood of a certain POVM element. Specifically, consider a local projective measurement $M=\{\ketbra{ij}\}$. We can coarse-grain this into a two-element projective measurement with the bipartition into ``agreeing'' outcomes, $E_A=\sum_i \ketbra{ii}$, and ``disagreeing'' outcomes, $E_D=\sum_{i\neq j}\ketbra{ij}$, respectively. 
Lastly, so that these outcomes are basis-independent, we can ``twirl'' $E_A$ and $E_D$ as follows:
\begin{eqnarray*}
\overline{E}_A=\int d\mu(U) U\otimes U \Big( \sum_i \ketbra{ii} \Big) U^{\dagger} \otimes U^{\dagger}\\
\overline{E}_D=\int d\mu(U) U\otimes U \Big( \sum_{i\neq j} \ketbra{ij} \Big) U^{\dagger} \otimes U^{\dagger}, 
\end{eqnarray*}
where $d\mu(U)$ denotes integration with respect to the invariant (Haar) measure. 
It is simple to see that these two operators yield a resolution of identity and hence form a POVM. 
We can compute these two operators explicitly as follows. By the invariance of the Haar measure, we can rewrite 
$\overline{E}_A$ as
\begin{eqnarray*}
 \overline{E}_A=d\int d\mu(\psi) \ketbra{\psi}^{\otimes 2},
\end{eqnarray*}
for which the above integral is proportional to the projector onto (or identity operator $\identity^+_2$ in) the totally symmetric subspace $\hilbert_+^{\otimes 2}\subset \hilbert^{\otimes 2}$ \cite{Chiribella2006}. Explicitly, we can write 
\begin{eqnarray}
\overline{E}_A &=& \frac{d}{d^+_2}\identity^+_2=\frac{d}{d^+_2}\frac{\identity+V}{2}, \quad 
 d_2^+ \equiv  \text{dim}(\hilbert^+_2)={2+d-1\choose 2}, 
\label{povm0}\\
\overline{E}_D & = & \identity - \overline{E}_A . 
\label{povm}
\end{eqnarray}
We define the {\em degree of agreement} to be the likelihood of $\overline{E}_A$ and, similarly, the {\em degree of disagreement} to be the likelihood of $\overline{E}_D$. Operationally, these values are the probability that, for a randomly chosen local projective measurement made collectively, the local outcomes will agree or disagree. 

We now proceed to show how quantum channels differ from quantum states in their allowed range of agreement likelihood. In the case of a bipartite operator $\rho\in\mathcal{B}(\hilbert\otimes\hilbert)$, we are familiar with computing this  agreement probability as $\tr{}({\overline{E}_A\rho})$. 
To carry out the same computation for a channel operator, the homocorrelation map becomes expedient. 
Given a quantum channel $\mathcal{M}:\mathcal{B}(\hilbert)\rightarrow\mathcal{B}(\hilbert)$, we wish to determine the probability the outcome of a randomly chosen projective measurement (made on the completely mixed state) will agree with the outcome of the same measurement after the application of $\mathcal{M}$. Assume the outcome was $\ket{i}$ from an orthogonal basis $\{\ket{i}\}$. Then the post-channel state is $\mathcal{M}(\ketbra{i})$, and the likelihood that the post-channel measurement will also be $\ket{i}$ is
$\bra{i}\mathcal{M}(\ketbra{i})\ket{i}$. 
Lastly, if we want to average this likelihood of agreement over all choices of basis we integrate,
\begin{eqnarray*}
\text{p(agree)}&=\int d\mu(U)\tr{}\Big( {\mathcal{M}(U\ketbra{i}U^{\dagger})U\ketbra{i}U^{\dagger}} \Big)\\
&=\int d\mu(\psi)\tr{}\Big( {\mathcal{M}(\ketbra{\psi})\ketbra{\psi}} \Big).
\end{eqnarray*}
If we wish to find the bounds on this value, the above form does not make transparent the fact that we are performing an optimization problem in a convex cone. But, recalling the namesake property of the homocorrelation map, Eq. (\ref{eq:homocorrprop}), the above expression may be rewritten as
\begin{eqnarray*}
\text{p(agree)}=\tr \left[ \mathcal{H}(\mathcal{M})d\int d\mu(\psi)\ketbra{\psi}\otimes\ketbra{\psi} \right] 
=\tr[ \mathcal{H}(\mathcal{M})\overline{E}_A].
\end{eqnarray*}
Accordingly, the likelihood of agreement is calculated for channel operators in the homocorrelation representation {\em in just the same way} as it is for bipartite density operators. 
With the stage set, the desired bounds are described in the following theorem:
\begin{thm}
Let $w$ be an operator in $\mathcal{B}(\hilbert_d\otimes\hilbert_d)$,
and consider a POVM with operation elements as in Eq.
(\ref{povm}). Then the degree of agreement for $w\geq_{\text{st}} 0$ is bounded by 
\begin{equation}\label{eq:b1}
0\leq \trn{}{(w\overline{E}_A)}\leq \frac{2}{2+d-1},
\end{equation}
while the degree of agreement for $w\geq_{\text{ch}}0$ is bounded by 
\begin{equation}\label{eq:b2}
\frac{1}{2+d-1}\leq \trn{}({w\overline{E}_A})\leq 1.
\end{equation}
\label{thm:bipartiteagree}
\end{thm}

\vspace*{-4mm}
\begin{proof}
In the case of state-positive operators, the maximal value of $\trn{}({w\overline{E}_A})$ is achieved 
by setting $w=\overline{E}_A/\trn{}{\overline{E}_A}$, which results in $\trn{}({w\overline{E}_A})={2}/({d+1})$. For the lower bound, it is simple to see that choosing $w$ to lie in the complement of the projector yields a value of zero. Hence, we have obtained the bound of Eq. (\ref{eq:b1}).

In the case of channel-positive operators, the value of $\trn{}(w\overline{E}_A)$, where $w\geq_{\text{ch}}0$, is unchanged by a partial transposition of both operators. Thus, we may seek bounds on the value of $\trn{}({w^{T_A} \overline{E}_A^{T_A}})$, where $w^{T_A}$ is a density operator. By using Eq. (\ref{povm0}), 
the partial transposition of $\overline{E}_A$ is
$$\overline{E}_A^{T_A}=\frac{d}{d^+_2}\frac{\identity+ V^{T_A}}{2}.$$
Thus, the upper and lower bounds on $\trn{}({w\overline{E}_A})$ are achieved by setting $w^{T_A}= V^{T_A}/d$ and $w^{T_A}=(\identity- V^{T_A})/(d^2-d)$, respectively. Accordingly, the resulting bounds are $\frac{d}{d^+_2}\frac{1}{2}\leq \trn{}({w\overline{E}_A})\leq \frac{d}{d^+_2}\frac{1+d}{2}$, which simplify to those of Eq. (\ref{eq:b2}).
\end{proof}
By virtue of the homocorrelation map, the above result may be understood geometrically. The objects involved are the agreement/disagreement POVM operators $\overline{E}_A$ and $\overline{E}_D$, and the state- and channel- positive cones $W_{\text{st}}$ and $W_{\text{ch}}$, respectively. Theorem \ref{thm:bipartiteagree} places an {\em upper bound} on the inner product between vectors in $W_{\text{st}}$ and $\overline{E}_A$, and, similarly, on the inner product between vectors in $W_{\text{ch}}$ and $\overline{E}_D$. This geometric understanding is aided by the example of Werner operators shown in Fig. \ref{fig:statechannelcones}.

Lastly, we proceed to show that general joinability limitations (though not strict ones) can be derived based solely on i) the above agreement bounds of channels and states; ii) joinability bounds of classical probabilities; and iii) the fact that the agreement likelihoods must obey rules of classical joinability. Ultimately, the reduced states must satisfy certain limitations arising from joining limitations of classical probability distributions. In the three-party joining scenario, the bipartite marginal distributions of three \emph{classical} $d$-nary random variables must have probabilities of agreement $\alpha_{AB}$, $\alpha_{AC}$, and $\alpha_{BC}$ satisfying the following inequalities \cite{Johnson2013}:
\begin{eqnarray}
-\alpha_{AB}+\alpha_{AC}+\alpha_{BC}&\leq &1,\\ 
\;\;\; \alpha_{AB}-\alpha_{AC}+\alpha_{BC}&\leq &1,\\
\;\;\; \alpha_{AB}+\alpha_{AC}-\alpha_{BC}&\leq &1,
\label{eq:channelbound}
\end{eqnarray}
and, in the case of $d=2$, also 
\begin{equation}\label{eq:statebound}
\;\;\; \alpha_{AB}+\alpha_{AC}+\alpha_{BC}\geq 1.
\end{equation} 
Since $\trn{}({w\overline{E}_A})$ is a probability of agreement, it too is subject to the above constraints. 
Hence, we identify $\trn{}({\rho_i\overline{E}_A}) \equiv \alpha_{\ell}$, where $\ell =AB$, $AC$, or $BC$. Consider the case where systems $B$-$C$ are state-positive. Theorem \ref{thm:bipartiteagree} then sets the bound $\trn{}({\rho_{BC}\overline{E}_A})\leq\frac{2}{d+1}$. Setting the parameter $\alpha_{BC}=\trn{}({\rho_{BC}\overline{E}_A})$ to this upper limit of $\frac{2}{d+1}$, Eq. (\ref{eq:channelbound}) becomes
$$ \alpha_{AB}+\alpha_{AC}\leq \frac{d+3}{d+1}.
$$
In the case of $\alpha_{AB}=\alpha_{AC}\equiv \alpha$, this yields
$$ \alpha\leq \frac{d+3}{2(d+1)},
$$
which corresponds precisely to the optimal bound for qudit cloning \cite{Scarani2005} (cf. Eq. (21) therein, where their $F$ coincides with our $\alpha$).
We can similarly recover the exact bound for the 1-2 sharability of qubit Werner states determined in \cite{Johnson2013}. Again, we set the $B$-$C$ agreement to its extremal value $\trn{}({\rho_{BC}\overline{E}_A})=\frac{2}{d+1}$, as given by Theorem \ref{thm:bipartiteagree}. For $d=2$, Eq. (\ref{eq:statebound}) applies, and substituting in the extremal value of $\alpha_{BC}$ we obtain $ \alpha_{AB}+\alpha_{AC}\geq \frac{1}{3}$.
Again, in the case of $\alpha_{AB}=\alpha_{AC} \equiv \alpha$, this yields 
$ \alpha\leq \frac{1}{6}$, which is the exact condition for 1-2 sharability of Werner qubits.

While obtaining a full generalization of Theorem \ref{thm:bipartiteagree} to multiparty systems would entail 
a detailed understanding of representation theory for Brauer algebras which is beyond our current purpose, 
we can nevertheless establish the following: 
\begin{thm}
Let $w \in \mathcal{B}(\hilbert_d^{\otimes N})$, and 
consider a POVM with operation elements $\overline{E}_A=\frac{d
}{d^+_N}\identity^+_{N}$ and $\overline{E}_D=\identity-\overline{E}_A$ (analogous to Eq. (\ref{povm})). Then the degree of agreement for $w\geq_{\text{st}} 0$ as calculated by the likelihood of $\overline{E}_A$ is bounded by 
\begin{equation}\label{eq:b3}
0\leq \trn{}({w\overline{E}_A}) \leq \frac{d} {{d-1+N \choose N}}. 
\end{equation}
\end{thm}
\begin{proof}
The maximal and minimal values of $\trn{}({w\overline{E}_A})$ are achieved by setting $w=\overline{E}_A/\trn{}{\overline{E}_A}$ and $w=(d\identity/d^+_N-\overline{E}_A)/\trn{}({d\identity/d^+_N) -\overline{E}_A}$, respectively, which yields the desired bounds of Eq. (\ref{eq:b3}).
\end{proof}
From the above multiparty bound, one may attempt to recover, for instance, the known bounds on 1-$n$ sharability of Werner states
\cite{Johnson2013}. However, we have, thus far, not been successful in this endeavor. In the tripartite qudit setting, such bounds were found to be sufficient, but this might be a special feature of this particular case. Therefore, it remains an open question to determine whether 
there exists a simple principle (or simple principles) which govern joinability limitations beyond the tripartite setting.

\section{Conclusion}
\label{sec:conclusion}

In this paper we have developed a unifying framework for the concept of quantum joinability. Many problems regarding the part-whole relationship in multiparty quantum settings, such as the quantum marginal problem, the asymmetric cloning problem, and various quantum extension problems, are  encapsulated by this framework. 
An important step was to introduce the \emph{homocorrelation map} as a natural way to represent quantum channels with bipartite operators, making them geometrically comparable to quantum states. Using this tool, it is possible 
to {\em directly} contrast the joinability properties of quantum states with those of quantum channels. 
In particular, applying the framework to the simplest case of $U\otimes U\otimes U$-invariant operators, we found that the state and channel joinability bounds work in tandem to exhibit the symmetry inherent in the limitations of classical joinability. In addition, we derived the local-positivity joinability bounds in this setting. Though less strict than state- or channel- joinability bounds, we found that the {\em local positivity joinability bounds are still more strict than purely classical ones}, and provided an operational interpretation of this fact. 

The Choi-Jamiolkowski ismorphism illuminates a duality between bipartite quantum states and quantum channels. As another main finding of this work, we have emphasized a crucial difference between the two, that manifests in the correlations that are obtainable from each. Namely, bipartite quantum states are limited in their agreement, whereas quantum channels are limited in their disagreement.  Again, this difference is made explicit by representing quantum channels with the homocorrelation map.
We showed how these differences, expressed in terms of agreement bounds, in turn 
inform the joinability properties of channels vs states.
In view of their general nature, these agreement bounds may have further implications yet to be discovered.

In closing, we note that throughout our analysis we have only considered scenarios with a pre-defined tensor product structure, and consequently all operator reductions are obtained via the usual partial-trace construction. However, it is important to appreciate that this was {\em not} a necessary restriction. Following \cite{Barnum2003}, one may also consider a more general notion of a reduced state, which results from appropriately restricting the global state to a distinguished operator subspace. 
Such a notion of reduction is operationally motivated in situations where a tensor product structure is not uniquely or naturally afforded on physical grounds (notably, systems of indistinguishable particles or operational quantum theory, see e.g.  \cite{Barnum2004}). This naturally points to a further extension of the present joinability framework ``beyond subsystems", 
which we plan to address in future investigation.

\section{Acknowledgements}

This work was inspired by discussions with Sandu Popescu, David Sicilia, Rob Spekkens, and Bill Wootters. Support from the Constance and Walter Burke Special Projects Fund in Quantum Information Science is gratefully acknowledged.

\appendix 
\section{Local positivity of Werner operators}

We present here a detailed proof of Thm. \ref{thm:locpos}. 
The first step is to show that considering tripartite joining state of a simpler form suffices in the qubit case.

An arbitrary tripartite Werner operator may be parametrized as 
\begin{eqnarray*}
\hspace*{-1cm}
w=a\identity+bV_{(AB)}+cV_{(AC)}+dV_{(BC)}
+e(V_{ABC}+V_{CBA})/2+if(V_{ABC}-V_{CBA})/2,
\end{eqnarray*}
where $a,b,c,d,e,f \in {\mathbb R}$ and normalization is left arbitrary for now. However, in the two-dimensional case, 
the six permutation representation operators are not independent, since $\identity-(V_{(AB)}+V_{(AC)}+V_{(AC)})+V_{(ABC)}+V_{(CBA)}=0$. Consequently, we may absorb the $V_{(ABC)}+V_{(CBA)}$ contribution into the first four terms, leaving us with
\begin{eqnarray*}
w=a\identity+bV_{(AB)}+cV_{(AC)}+dV_{(BC)}
+if(V_{ABC}-V_{CBA})/2.
\end{eqnarray*}
With $\ket{\psi_{\text{loc}}}\equiv\ket{\psi_1}\ket{\psi_2}\ket{\psi_3}$, local positivity of $w$ is guaranteed by $\bra{\psi_{\text{loc}}}w\ket{\psi_{\text{loc}}}\geq0$,  
 holding for all $\ket{\psi_1},\ket{\psi_2},\ket{\psi_3}\in\hilbert$. Writing
\begin{eqnarray*}
\bra{\psi_{\text{loc}}}w\ket{\psi_{\text{loc}}}&=a +b|\inprod{\psi_1}{\psi_2}|^2+c|\inprod{\psi_1}{\psi_3}|^2+d|\inprod{\psi_2}{\psi_3}|^2
\\&+f\text{Im}(\inprod{\psi_1}{\psi_2}\inprod{\psi_2}{\psi_3}\inprod{\psi_3}{\psi_1})\geq0,
\end{eqnarray*}
each choice of $\ket{\psi_1},\ket{\psi_2},\ket{\psi_3}$ enforces a linear inequality on $a,b,c,d,e,f$. However, 
certain $\ket{\psi_1},\ket{\psi_2},\ket{\psi_3}$ may result in an inequality whose satisfaction is guaranteed by a stricter inequality corresponding to a different set of product vectors. For each choice of $a,\ldots,f$, there will be an extremal (set of) product vector(s) $\ket{\psi'_{\text{loc}}}$ for which $\bra{\psi'_{\text{loc}}}w\ket{\psi'_{\text{loc}}}\geq0$ implies $\bra{\psi_{\text{loc}}}w\ket{\psi_{\text{loc}}}\geq0$ for all $\ket{\psi_{\text{loc}}}$. We seek to obtain such extremal product vectors, and write their inner products (e.g. $|\inprod{\psi_1}{\psi_2}|^2$, etc.) in terms of $a,\ldots,f$.

For Werner states, the local-positivity condition is invariant under a collective unitary transformation of $\ket{\psi_{\text{loc}}}$. Such a transformation corresponds to a rotation on the Bloch sphere. Thus, given $\ket{\psi_1}\ket{\psi_2}\ket{\psi_3}$, we may perform a collective unitary which takes this state to $\ket{\uparrow_z}\otimes(\cos{\theta}\ket{\uparrow_z}+\sin{\theta}\ket{\downarrow_z})\otimes(\cos{\Omega}\ket{\uparrow_z}+e^{i\phi}\sin{\Omega}\ket{\downarrow_z}).$ Without loss of generality, this will be our representative $\ket{\psi_{\text{loc}}}$.
This allows us to rewrite the expression of local-positivity as
%
\begin{eqnarray*}
\hspace*{-2cm} F &=& a +b\cos^2{\theta}+c\cos^2{\Omega}\\
\hspace*{-2cm} &+&d(\cos^2\theta\cos^2\Omega+\frac{1}{2}\cos\phi\sin2\theta\sin2\Omega+\sin^2\theta\sin^2\Omega) 
+\frac{f}{4}\sin{\phi}\sin{2\theta}\sin{2\Omega}\geq 0.
\end{eqnarray*}
Our goal is to determine the set of bipartite Werner operator trios that can be joined by a local-positive state $w$. These reduced states on $A$-$B$, $A$-$C$, and $B$-$C$ are each characterized by the single parameter $\alpha_{AB}=\tr(V_{(AB)}w)$, $\alpha_{AC}=\tr(V_{(AC)}w)$, and $\alpha_{BC}=\tr(V_{(BC)}w)$, respectively. In the next step, we show that if the local-positive state $w$ joins reduced Werner states with $\alpha_{AB}$, $\alpha_{AC}$, and $\alpha_{BC}$, then $w'=w|_{f=0}$ is local-positive and also joins them.

First, note that the bipartite reduced states $\tr_C(w)$, etc., do not depend on $f$; hence, if three bipartite states are local-positive-joinable by some $w$ with $f\neq0$, then $w'=w|_{f=0}$ will reduce to the same bipartite states as $w$. It remains to show that $w'$ is local-positive. Specifically, we want to show that if $F\geq 0$ for all $\theta, \Omega, \phi$, then $F(f=0)\geq 0$ for all $\theta, \Omega, \phi$. This follows from the fact that, independent of all else, the factor of $\sin\phi$ may determine the sign of its corresponding term; thus, for a given $a,\ldots,f$, the angles which minimize $F$ must be such that the term containing $f$ is non-positive. In this case, setting $f=0$ cannot decrease $F$.

We have thus shown that a sufficient joining state is of the form
\begin{eqnarray*}
w=a\identity+bV_{(AB)}+cV_{(AC)}+dV_{(BC)},
\end{eqnarray*}
and, in terms of the parameterization of the product state $\ket{\psi_{\text{loc}}}$, local positivity is ensured by requiring that
\begin{eqnarray*}
F= a +b\cos^2{\theta}+c\cos^2{\Omega}\\
+d (\cos^2\theta\cos^2\Omega+\frac{1}{2}\cos\phi\sin2\theta\sin2\Omega+\sin^2\theta\sin^2\Omega)
\geq 0, 
\end{eqnarray*}
for all $\theta,\Omega\in[0,\pi],\,\,\phi\in[0,2\pi]$. It remains to determine the extremal angles $\theta$, $\Omega$, and $\phi$, for a given $a,b$, $c,d$.
With respect to the $\phi$ dependence, $F$ is extremized by setting $\cos\phi=\pm1$, which determines the sign of the corresponding term. However, the sign of this term is also determined by the sign of $\theta$ or $\Omega$, which does not alter the remainder of the expression for $F$. Thus, we absorb this choice of $\cos\phi=\pm1$ into the sign of $\theta$, say. This allows us to further simplify our expression to
$$F=a +b\cos^2{\theta}+c\cos^2{\Omega}+d\cos^2{(\theta-\Omega)}.
$$ 
The interpretation of this simplification is that it suffices to consider states $\ket{\psi_1},\ket{\psi_2},\ket{\psi_3}$ all lying in an equatorial plane of the Bloch sphere.
%
We make a final simplification by enforcing the normalization $\tr(w)=1$. This removes $a$ 
by $a=\frac{1}{8}-\frac{1}{2}(b+c+d)$, giving
\begin{equation}\label{eq:angleparam}
F=\frac{1}{2}\Big(\frac{1}{4}+b\cos{\theta}+c\cos{\Omega}+d\cos{(\theta+\Omega)}\Big),
\end{equation}
where we have replaced $2\theta\rightarrow\theta$ and $-2\Omega\rightarrow\Omega$ without loss of generality.

Now in order to find the desired extremal inequalities, we take partial derivatives with respect to the remaining 
two angles, namely: 
\begin{eqnarray*}
\hspace*{-5mm}\frac{\partial F}{\partial\theta}=-b\sin{(\theta)}-d\sin{(\theta+\Omega)}=0, \quad
\frac{\partial F}{\partial\Omega}=-c\sin{(\Omega)}-d\sin{(\theta+\Omega)}=0.
\end{eqnarray*}
Assuming $b,c,d,\neq0$, the zeros of the gradient of $F$ are given by either
$$ \sin\theta=\sin\Omega=\sin(\theta+\Omega)=0, 
$$ 
or
$$ 
 \frac{b}{d}=-\frac{\sin(\theta+\Omega)}{\sin\theta},\qquad\frac{c}{d}=-\frac{\sin(\theta+\Omega)}{\sin\Omega}.
$$ 
The first set of solutions correspond to $\theta=n\pi$ and $\Omega=m\pi$. There are four inequalities derived from these
\begin{eqnarray}
\label{eq:linearlocposbounds1}
\frac{1}{4}+b+c+d\geq0,\quad  
\frac{1}{4}+b-c-d\geq0,\\
\frac{1}{4}-b+c-d\geq0,\quad
\frac{1}{4}-b-c+d\geq0.
\label{eq:linearlocposbounds2}
\end{eqnarray}
Satisfaction of these is certainly necessary for $w$ to be locally positive, but it is not sufficient.

Although they do not minimize $F$ for all $b,c,d$, the solutions $\sin\theta=\sin\Omega=\sin(\theta+\Omega)=0$ allow us to obtain four equalities
\begin{eqnarray*}
\cos{x}=\frac{\cos{x}\sin{y}}{\sin{y}}=\frac{\sin{(x+y)}-\sin{(x-y)}}{2\sin{y}},\\
\sin{(x+y)}\sin{(x-y)}=\sin^2{x}-\sin^2{y}.
\end{eqnarray*}
Putting these together we have
\begin{eqnarray*}
\cos{x}=\frac{1}{2}\left[\frac{\sin{(x+y)}}{\sin{y}}+\frac{\sin{y}}{\sin{(x+y)}}
-\frac{\sin^2{x}}{\sin{y}\sin{(x+y)}}\right].
\end{eqnarray*}
Thus, we can write each of the $\cos$ terms in terms of $b,c,d$
\begin{eqnarray*}
\cos{\theta}=\frac{1}{2}\left[\frac{c}{d}+\frac{d}{c}-\frac{cd}{b^2}\right],\\
\cos{\Omega}=\frac{1}{2}\left[\frac{b}{d}+\frac{d}{b}-\frac{bd}{c^2}\right],\\
\cos{(\theta+\Omega)=-\frac{1}{2}\left[\frac{b}{c}+\frac{c}{b}-\frac{bc}{d^2}\right]}.
\end{eqnarray*}
Substituting these into Eq. (\ref{eq:angleparam}), we obtain 
\begin{equation}\label{eq:last}
\frac{1}{2}-\frac{(bc)^2+(bd)^2+(cd)^2}{bcd}\geq 0
\end{equation}
as the remaining necessary condition for local positivity. The above condition, along with Eqs. (\ref{eq:linearlocposbounds1})-(\ref{eq:linearlocposbounds2}) ensure the local positivity of the relevant states.
As a final step, note that the Werner parameters of Eq. (\ref{eq:Werner}) are related to $b,c,d$ via 
$$b=\frac{1}{4}\eta_{AB}, \quad c=\frac{1}{4}\eta_{AC}, \quad 
d=\frac{1}{4}\eta_{BC}.$$ 
Upon re-expressing Eqs. (\ref{eq:linearlocposbounds1}), (\ref{eq:linearlocposbounds2}), and (\ref{eq:last}) 
in terms of Werner parameters $\eta$s, the result quoted in Thm. \ref{thm:locpos} is established. \qed

\section*{References}



\providecommand{\newblock}{}

\end{document}